\documentclass[12pt]{article}
\usepackage{amssymb}
\usepackage{amsmath}
\usepackage{amsmath}
\usepackage{amsthm}
\usepackage{amsfonts}
\usepackage{graphicx}
\usepackage{enumerate}
\usepackage{bbm} 
\usepackage{dsfont}
\usepackage[authoryear,longnamesfirst]{natbib}
\usepackage{multirow}

\numberwithin{equation}{section}
\newtheorem{theorem}{Theorem}

\newtheorem{assumption}{Assumption}

\newtheorem{corollary}{Corollary}

\newtheorem{proposition}{Proposition}
\newtheorem{remark}{Remark}

\newtheorem{lemma}{Lemma}
\theoremstyle{definition}

\newcommand{\norm}[1]{\left \|#1\right \|}
\newcommand{\abs}[1]{\left\vert #1 \right\vert}

\newcommand{\supp}{\text{supp}}

\newcommand{\sumi}{\sum_{i=1}^N }
\newcommand{\sumj}{\sum_{j=1}^M }
\newcommand{\sumt}{\sum_{t=1}^T }

\begin{document}

\title{Post-Selection Inference in Three-Dimensional Panel Data\thanks{First arXiv version: March 30, 2019.}}
\author{
Harold D. Chiang\thanks{Harold D. Chiang: harold.d.chiang@vanderbilt.edu. Department of Economics, Vanderbilt University, VU Station B \#351819, 2301 Vanderbilt Place, Nashville, TN 37235-1819, USA\bigskip} \qquad Joel Rodrigue\thanks{Joel Rodrigue: joel.b.rodrigue@vanderbilt.edu. Department of Economics, Vanderbilt University, VU Station B \#351819, 2301 Vanderbilt Place, Nashville, TN 37235-1819, USA\bigskip} \qquad Yuya Sasaki\thanks{Yuya Sasaki: yuya.sasaki@vanderbilt.edu. Department of Economics, Vanderbilt University, VU Station B \#351819, 2301 Vanderbilt Place, Nashville, TN 37235-1819, USA\bigskip}
\bigskip\\
}
\date{}
\maketitle
\begin{abstract}\setlength{\baselineskip}{5.8mm}
Three-dimensional panel models are widely used in empirical analysis.
Researchers use various combinations of fixed effects for three-dimensional panels.
When one imposes a parsimonious model and the true model is rich, then it incurs mis-specification biases.
When one employs a rich model and the true model is parsimonious, then it incurs larger standard errors than necessary.
It is therefore useful for researchers to know correct models.
In this light, \citet{LuMiaoSu2018} propose methods of model selection.
We advance this literature by proposing a method of post-selection inference for regression parameters.
Despite our use of the lasso technique as means of model selection, our assumptions allow for many and even all fixed effects to be nonzero.
Simulation studies demonstrate that the proposed method is less biased than under-fitting fixed effect estimators, is more efficient than over-fitting fixed effect estimators, and allows for as accurate inference as the oracle estimator.
\end{abstract}

\newpage
\section{Introduction}
\label{sec:introduction}

\citet{Matyas1997} suggests the three-dimensional panel model
\begin{align}\label{eq:matyas}
y_{ijt} = {x}_{ijt}' \beta + \underbrace{\alpha_i + \gamma_j + \lambda_t}_{\text{Fixed Effects}} + \varepsilon_{ijt}
\end{align}
for $(i,j,t) \in \{1,...,N\} \times \{1,...,M\} \times \{0,...,T\}$, where $y_{ijt}$ denotes an outcome variable of unit $(i,j)$ at time $t$, ${x}_{ijt}$ denotes $k$-dimensional explanatory variables of unit $(i,j)$ at time $t$, and $\alpha_i$, $\gamma_j$, and $\lambda_t$ are fixed effects associated with indices $i$, $j$, and $t$, respectively.
To fix our ideas, consider the gravity model \citep{Tinbergen1962} from the empirical trade literature where $y_{ijt}$ denotes the logarithm of the volume of exports from country $i$ to country $j$ in year $t$, and
the $k$-dimensional covariates ${x}_{ijt}$ contain observed characteristics of the trade pair $(i,j)$ in year $t$, including the log GDP of country $i$ in year $t$ ($GDP_{it}$), the log GDP of country $j$ in year $t$ ($GDP_{jt}$), the log distance between countries $i$ and $j$ ($DIST_{ij}$), and the dummy variable of a bilateral trade agreement between countries $i$ and $j$ ($TA_{ij}$), among others.
The fixed effects $\alpha_i$, $\gamma_j$, and $\lambda_t$ represent the unobserved exporting country effects, destination country effects, and year effects, respectively.
Researchers are often interested in the coefficient of $DIST_{ij}$ interpreted as the trade elasticity or the trade cost.
Another important parameter of empirical interest is the coefficient of $TA_{ij}$ interpreted as the effect of bilateral trade agreements on trade volumes. 
See \cite{HeadMayer2014} for a comprehensive review of gravity models.

To date, variants of the three-dimensional panel model (\ref{eq:matyas}) have been extensively used in empirical analysis of international trade (see \citet{BaltagiEggerErhardt2017} for a survey), 
housing (see \citet{BaltagiBresson2017} for a survey), 
migration (see \citet{Ramos2017} for a survey), and 
consumer price.
In these analyses, researchers employ various combinations of fixed effects, including 
(I) $\alpha_i + \gamma_j$, 
(II) $\alpha_i + \gamma_j + \lambda_t$, and
(III) $\alpha_{it} + \gamma_{jt}$,
among others.\footnote{
Parameters $\beta$ of certain types of controls are not identified under more general combinations of fixed effects.
For example, the coefficients of $GDP_{it}$ and $GDP_{jt}$ are not identified under the fixed effect model (III) due to the collinearity.
However, the coefficients of $DIST_{ij}$ and $TA_{ij}$ would be identifiable under any of the three models.
In empirical analysis of bilateral trade flows, the latter two coefficients are of more common interest.
In fact substituting fixed effects (such as $\alpha_{it}$ and $\gamma_{jt}$) for observed proxies (such as $GDP_{it}$ and $GDP_{jt}$) is ``now common practice and recommended by major empirical trade economists'' \citep{HeadMayer2014}.
}
See \citet[][Tables 1.1--1.3]{BalazsiMatyasWansbeek2017} for a comprehensive list of empirical papers and their specifications of the combinations of fixed effects.
Researchers in general do not know which combination of fixed effects correctly specifies the model of their interest.
If the true model is parsimonious and a researcher erroneously assumes a rich specification, then na\"ive fixed effect estimators generally entail exacerbated variances.
On the other hand, if the true model is rich and a researcher erroneously assumes a parsimonious specification, then na\"ive fixed effect estimators generally entail mis-specification biases.
The lack of knowledge of the true model specification therefore leads to undesired econometric results in any event.

A recent paper by \citet{LuMiaoSu2018} develops a method of model selection.
Their method serves as a useful guideline for empirical researchers to choose a correct combination of fixed effects in three-dimensional panel models.
When a researcher uses a selected model to compute estimates of $\beta$ and their standard errors, it is also important that she takes into account the statistical effects of the model selection.
To our knowledge, the existing literature does not provide a method of post-selection inference for three-way panel models.
In this light, we extend the frontier of this existing econometric literature \citep{LuMiaoSu2018} by providing a method of inference for $\beta$ accounting for the effect of the model selection.
We make use of the lasso technique along with de-biasing to this end, but our method does not require exactly sparse fixed effects.
In other words, our assumptions do allow for many and even all of the fixed effects to be nonzero in a general combination of fixed effects.

{\bf Related Literature}
A three-dimensional panel model was suggested by \citet{Matyas1997}.
The literature on multi-dimensional panels is extensive today, and is surveyed in the book of article collections edited by \citet{Matyas2017}.
Its chapter written by \citet{BalazsiMatyasWansbeek2017} provides a comprehensive list of empirical research papers employing multi-dimensional panel data.

Methods of model selection in three-dimensional panels are developed by \citet{LuMiaoSu2018}, and
this paper was motivated by \citet{LuMiaoSu2018}.
As stated earlier, we aim to extend this frontier of the literature by developing a post-selection inference for the regression parameters.

We use the lasso technique for model selection and post-selection inference, but our assumptions do allow for all fixed effects to be nonzero.
This is because we rely on the approximate sparsity condition as opposed to the conventional sparsity.
Post-selection inference via lasso is studied by an extensive body of the literature in various contexts.
This literature includes, but are not limited to, 
\citet{BelloniChenChernozhukovHansen2012} for IV models, and \citet{BelloniChernozhukovHansen2014}, \citet{JavanmardMontanari2014}, \citet{vandeGeeretal2014}, and \citet{ZhangZhang2014} for linear regression models.

Lasso estimation for panel models are suggested by 
\citet{Koenker2004}, 
\citet{Lamarche2010}, 
\citet{Kock2013}, 
\citet{CanerHan2014}, 
\citet{LuSu2016}, 
\citet{LiQianSu2016}, 
\citet{QianSu2016}, 
\citet{CanerHanLee2018}, 
\citet{HardingLamarche2019}, among others.
Classification and estimation by lasso for panel models are proposed by \citet{SuShiPhillips2016} -- also see \citet{LuSu2017}, \citet{SuJu2018}, and \citet{SuWangJin2017}. 
For post-selection inference with panel data using lasso,
\citet{BelloniChernozhukovHansenKozbur2016} work with de-meaned fixed effect models with high-dimensional controls using post-double-selection estimator.
\citet{Kock2016} and \citet{KockTang2018} work with correlated random effect panel models and dynamic panel models with sparse fixed effects via de-biased lasso, respectively. 
We extend this frontier of the literature to three-dimensional panels.
Besides the different framework of three-dimensional panels as opposed to two-dimensional ones, this paper is different from \citet{Kock2016} and \citet{KockTang2018} in the following four technical points.
First, 
we extend the theory of nodewise lasso by allowing for different convergence rates to incorporate a larger class of fixed effect models. 
Second, we use a different proof strategy with the sparsity requirement of $ss_l(\log (p\vee (NM)))^2/(N \wedge M)=o(1)$ inspired by \citet[][Lemma 8]{BelloniChenChernozhukovHansen2012}, whereas an adaptation of the proof strategies of \citet{Kock2016}\footnote{See Assumption A3 (b) of \citet{Kock2016}.} and \citet{KockTang2018}\footnote{See Assumption 5 (c) of \citet{KockTang2018}.} to our framework would require $ss^2_l (\log(p\vee (NM)))^2/(N \wedge M)=o(1)$.
This feature further extends the class of models that can be handled under our framework. 
Third, the sub-gaussianity assumption of covariates, which is assumed by the majority of papers in the de-biased lasso literature, is not required. 
Fourth, we allow for non-sparse coefficients based on the notion of approximate sparsity following that of \citet{BelloniChenChernozhukovHansen2012} instead of the $L^v$ sparsity for $0<v<1$ as in \citet{KockTang2018}.

With all these technical relations to the existing literature, we once again emphasize that our main contribution is the robust inference method for three-dimensional panels.
Unlike two-dimensional panels, there are a number of alternative combinations of fixed effect specifications in three-dimensional panels, and hence model selection is more important in these models \citep{LuMiaoSu2018}.
We apply and extend state-of-the-art technology \citep[e.g.,][]{BelloniChenChernozhukovHansen2012,Kock2016,KockTang2018} to this three-dimensional panel framework which concerns many empirical researchers.

{\bf Organization:}
The rest of this paper is organized as follows.
We introduce the model framework in Section \ref{sec:model}.
An overview of our proposed method is presented in Section \ref{sec:method}.
The main theoretical result is presented in Section \ref{sec:theory}, followed by sufficient conditions discussed in Section \ref{sec:sufficient_conditions_and_asymptotic_variance}.
We discuss the key assumption in the context of gravity analysis of international trade in Section \ref{sec:approximate_sparsity_in_gravity_analysis}.
We conduct simulation studies in Section \ref{sec:simulation_studies}.
Section \ref{sec:discussions} concludes the paper.

\section{The Model Framework}
\label{sec:model}

Consider the following representation of a general class of three-dimensional panel models with large $N$ and large $M$.
\begin{align}
y_{ijt} = {x}_{ijt}' \beta 
&
+ \sum_{i'=1}^{N} \alpha_{i'} \mathbbm{1}_{i=i'}
+ \sum_{j'=1}^{M} \gamma_{j'} \mathbbm{1}_{j=j'}
+ \sum_{t'=1}^{T} \lambda_{t'} \mathbbm{1}_{t=t'}
\notag\\
&
+ \sum_{i'=1}^{N} \sum_{t'=1}^{T} \alpha_{i't'} \mathbbm{1}_{i=i'} \mathbbm{1}_{t=t'}
+ \sum_{j'=1}^{N} \sum_{t'=1}^{T} \gamma_{j't'} \mathbbm{1}_{i=i'} \mathbbm{1}_{t=t'}
+ \varepsilon_{ijt}
\label{eq:three_way_representation}
\end{align}
This representation consists of a
$k$-dimensional parameter vector $\beta$, 
$N$-dimensional parameter vector $\alpha_{[N]}=(\alpha_1,...,\alpha_{N})'$,
$M$-dimensional parameter vector $\gamma_{[M]}=(\gamma_1,...,\gamma_{M})'$,
$T$-dimensional parameter vector $\lambda_{[T]}=(\lambda_1,...,\lambda_{T})'$,
$NT$-dimensional parameter vector $\alpha_{[NT]}=(\alpha_{11},...,\alpha_{NT})'$, and
$MT$-dimensional parameter vector $\gamma_{[MT]}=(\gamma_{11},...,\gamma_{MT})'$.
In total, there are $k+N+M+T+NT+MT$ parameters involved in this representation (\ref{eq:three_way_representation}).

Recall that conventional fixed effect models include
\begin{enumerate}[(I)]
	\item $\alpha_i + \gamma_j$,
	\item $\alpha_i + \gamma_j + \lambda_t$, and
	\item $\alpha_{it} + \gamma_{jt}$,
\end{enumerate}
among others.
Model (I) entails $k+N+M$ of possibly nonzero parameters $(\beta',\alpha_{[N]}',\gamma_{[M]}')'$, while the rest of the $T+NT+MT$ parameters $(\lambda_{[T]}',\alpha_{[NT]}',\gamma_{[MT]}')'$ are all zero.
Similarly, Model (II) entails $k+N+M+T$ of possibly nonzero parameters $(\beta',\alpha_{[N]}',\gamma_{[M]}',\lambda_{[T]})'$, while the rest of the $NT+MT$ parameters $(\alpha_{[NT]}',\gamma_{[MT]}')'$ are all zero.
Likewise, Model (III) entails $k+NT+MT$ of possibly nonzero parameters $(\beta',\alpha_{[NT]}',\gamma_{[MT]}')'$, while the rest of the $N+M+T$ parameters $(\alpha_{[N]}',\gamma_{[M]}',\lambda_{[T]}')'$ are all zero.
Furthermore, the representation (\ref{eq:three_way_representation}) includes many other combinations than these three models.

When Model (I) is true for example, then the representation (\ref{eq:three_way_representation}) has $T+NT+MT$ redundant parameters and hence estimating the model (\ref{eq:three_way_representation}) generally yields much larger standard errors for the parameters $\beta$ of interest than necessary.
This motivates the need of model selection.
We propose to use the lasso to select such redundant fixed effect parameters out of the representation (\ref{eq:three_way_representation}), and then conduct inference robustly accounting for the statistical effects of the model selection.

For ease of conducting econometric analysis, we further rewrite the representation (\ref{eq:three_way_representation}) as
\begin{align}
\label{eq:intermediate_representation}
y_{ijt} 
&= \mathbf{x}_{ijt}' \boldsymbol{\beta} 
+ \mathbf{d}_{1,it}' \boldsymbol{\alpha}
+ \mathbf{d}_{2,jt}' \boldsymbol{\gamma}
+ \varepsilon_{ijt}
\end{align}
where 
$\mathbf{x}_{ijt} = ({x}_{ijt}', \mathbbm{1}_{t=1}, ..., \mathbbm{1}_{t=T})'$ and 
$\boldsymbol{\beta} = (\beta', \lambda_1 , ..., \lambda_T)'$ are of dimension $k_0=k+T$,
$\mathbf{d}_{1,it} = (\mathbbm{1}_{i=1}, ..., \mathbbm{1}_{i=N}, \mathbbm{1}_{i=1}\mathbbm{1}_{t=1}, ..., \mathbbm{1}_{i=N}\mathbbm{1}_{t=T})'$ and
$\boldsymbol{\alpha} = (\alpha_{[N]},\alpha_{[NT]})'$ are of dimension $N_0=N+NT$, and
$\mathbf{d}_{2,jt} = (\mathbbm{1}_{j=1}, ..., \mathbbm{1}_{j=N}, \mathbbm{1}_{j=1}\mathbbm{1}_{t=1}, ..., \mathbbm{1}_{j=M}\mathbbm{1}_{t=T})'$ and
$\boldsymbol{\gamma} = (\gamma_{[M]},\gamma_{[MT]})'$ are of dimension $M_0=M+MT$.

Suppose that we can decompose the fixed effects $\boldsymbol{\alpha}$ into $\overline{\boldsymbol{\alpha}}$ and $\boldsymbol{\alpha} - \overline{\boldsymbol{\alpha}}$ and decompose the fixed effects $\boldsymbol{\gamma}$ into $\overline{\boldsymbol{\gamma}}$ and $\boldsymbol{\gamma} - \overline{\boldsymbol{\gamma}}$ such that 
\begin{align}\label{eq:alpha_decomposition}
\|\overline{\boldsymbol{\alpha}}\| \text{ is bounded and }
\sum_{i=1}^N \sum_{j=1}^M \sum_{t=1}^T \left( \mathbf{d}_{1,it}' \left(\boldsymbol{\alpha} - \overline{\boldsymbol{\alpha}}\right) \right)^2 \lesssim \norm{\boldsymbol{\beta}}_0 + \norm{\overline{\boldsymbol{\alpha}}}_0 + \norm{\overline{\boldsymbol{\gamma}}}_0
\end{align}
and
\begin{align}\label{eq:gamma_decomposition}
\|\overline{\boldsymbol{\gamma}}\| \text{ is bounded and }
\sum_{i=1}^N \sum_{j=1}^M \sum_{t=1}^T \left( \mathbf{d}_{2,it}' \left(\boldsymbol{\gamma} - \overline{\boldsymbol{\gamma}}\right) \right)^2 \lesssim \norm{\boldsymbol{\beta}}_0 + \norm{\overline{\boldsymbol{\alpha}}}_0 + \norm{\overline{\boldsymbol{\gamma}}}_0
\end{align}
hold,
where $\norm{\cdot}_0$ denotes the support cardinality (the $L^0$ norm).\footnote{With this said, we emphasize that this decomposition is merely theoretical, and a researcher need not implement such a decomposition in practice.
Precise requirements for the decomposition are stated in Assumptions \ref{a:sparsity}  and \ref{a:Theta} (4) ahead, followed by discussions in the context of our motivating application (\ref{eq:matyas}) in Remark \ref{remark:approximate_sparsity}.
In Section \ref{sec:approximate_sparsity_in_gravity_analysis}, we use world trade data to argue that these assumptions are plausible in the application (\ref{eq:matyas}).
}
Such a decomposition is constructed for example by setting $\overline{\boldsymbol{\alpha}}_\ell$ equal to $\boldsymbol{\alpha}_\ell$ for those coordinates $\ell$ for which $\abs{\boldsymbol{\alpha}_\ell}$ is large and setting $\overline{\boldsymbol{\alpha}}_\ell$ equal to zero for those coordinates $\ell$ for which $\abs{\boldsymbol{\alpha}_\ell}$ is small, and similarly for $\boldsymbol{\gamma}$.
Consequently, we can further rewrite the representation (\ref{eq:intermediate_representation}) as
\begin{align*}
y_{ijt} 
&= \mathbf{x}_{ijt}' \boldsymbol{\beta} 
+ \mathbf{d}_{1,it}' \overline{\boldsymbol{\alpha}}
+ \mathbf{d}_{2,jt}' \overline{\boldsymbol{\gamma}}
+ r_{ijt}
+ \varepsilon_{ijt}
\end{align*}
where $r_{ijt}$ is the approximation error defined by
\begin{align*}
r_{ijt} 
= \mathbf{d}_{1,it}' \left(\boldsymbol{\alpha} - \overline{\boldsymbol{\alpha}}\right)
+ \mathbf{d}_{2,jt}' \left(\boldsymbol{\gamma} - \overline{\boldsymbol{\gamma}}\right)
\end{align*}
and it satisfies
\begin{align*}
\sum_{i=1}^N \sum_{j=1}^M \sum_{t=1}^T r_{ijt}^2 \lesssim \norm{\boldsymbol{\beta}}_0 + \norm{\overline{\boldsymbol{\alpha}}}_0 + \norm{\overline{\boldsymbol{\gamma}}}_0.
\end{align*}

Stacking the three-dimensional panel data across the $NMT$ observations, we in turn construct the matrix representation
\begin{align}
Y 
= X\boldsymbol{\beta} + D_1 \overline{\boldsymbol{\alpha}} + D_2 \overline{\boldsymbol{\gamma}} + R + \varepsilon 
= Z\overline{\boldsymbol{\eta}} + R + \varepsilon,
\label{eq:representation}
\end{align}
where 
$Y=(y_{111},...,y_{NMT})'$,
$R=(r_{111},...,r_{NMT})'$, and
$\varepsilon=(\varepsilon_{111},...,\varepsilon_{NMT})'$,
are vectors of dimension $NMT$,
$X=(\mathbf{x}_{111},...,\mathbf{x}_{NMT})'$ is a matrix of size $NMT \times k_0$,
$D_1=(\mathbf{d}_{1,11},...,\mathbf{d}_{1,NT})'$ is a matrix of size $NMT \times N_0$,
$D_2=(\mathbf{d}_{2,11},...,\mathbf{d}_{2,MT})'$ is a matrix of size $NMT \times M_0$,
$Z=[X \ D_1 \ D_2]$, and
$\overline{\boldsymbol{\eta}} = [\boldsymbol{\beta}' \ \overline{\boldsymbol{\alpha}}' \ \overline{\boldsymbol{\gamma}}']'$ is a vector of dimension $p=k_0+N_0+M_0$.

If the true model is parsimonious, like Model (I), then a large number of the elements of the high-dimensional parameters, $\boldsymbol{\alpha}$ and $\boldsymbol{\gamma}$, will be zero.
Thus, a large number of the elements of $\overline{\boldsymbol{\alpha}}$ and $\overline{\boldsymbol{\gamma}}$ will be zero.
Furthermore, for those coordinates of $\boldsymbol{\alpha}$ and $\boldsymbol{\gamma}$ that are small in absolute value, the corresponding coordinates of $\overline{\boldsymbol{\alpha}}$ and $\overline{\boldsymbol{\gamma}}$ are set to zero in the decomposition in light of the relatively smaller approximation errors caused by setting them to zero.
We propose to use the lasso technique to select such redundant parameters in $\overline{\boldsymbol{\alpha}}$ and $\overline{\boldsymbol{\gamma}}$ out of this high-dimensional model as means of model selection for the purpose of obtaining smaller standard errors.
Furthermore, accounting for the statistical effects of this model selection, we then conduct robust inference for the main parameters $\boldsymbol{\beta}$ in the panel model.
Section \ref{sec:method} illustrates an overview of our proposed method.
A formal theoretical analysis will then follow in Sections \ref{sec:theory} and \ref{sec:sufficient_conditions_and_asymptotic_variance}.

\section{Overview of the Method}
\label{sec:method}

Our proposed method consists of four steps.
The first step is a lasso estimation of the parameter vector $\boldsymbol{\eta}$ entailing a model selection.
The second step is an auxiliary step to calculate an approximate inverse of the Gram matrix to be used in the subsequent two steps.
The third step de-biases the regularized lasso estimate from the first step.
The fourth step is a calculation of the asymptotic variance of each coordinate of the de-biased lasso estimator of $\boldsymbol{\beta}$.
\bigskip\\
{\bf Step 1:}
For the representing equation (\ref{eq:representation}), define the lasso estimator
\begin{align}\label{eq:lasso}
&\widehat{\boldsymbol{\eta}} \in \arg\min_{\boldsymbol{\eta} \in \mathbb{R}^p} \norm{Y-Z\boldsymbol{\eta}} + \mu P(\boldsymbol{\eta}),
\end{align}
where $\mu \in [0,\infty)$ is a regularization tuning parameter and the penalty function $P$ is defined by
\begin{align*}
&P(\boldsymbol{\eta}) = \norm{\widehat\Upsilon_1 \boldsymbol{\beta}}_1 + \frac{1}{\sqrt{N}} \norm{\widehat\Upsilon_2 \boldsymbol{\alpha}}_1 + \frac{1}{\sqrt{M}} \norm{\widehat\Upsilon_3 \boldsymbol{\gamma}}_1
\end{align*}
for some diagonal normalization matrix $\widehat\Upsilon_\ell$ for each $\ell \in \{1,2,3\}$.\footnote{See Remark \ref{rem:weights} in Appendix \ref{sec:oracle_inequalities}.}
In practice, the regularization tuning parameter $\mu$ can be chosen using a cross validation via software packages.
\bigskip\\
{\bf Step 2:}
The next step is an auxiliary process to obtain a $p \times p$ matrix $\widehat\Theta$ of approximate inverse of the Gram matrix to be used in Step 3.
We define the nodewise lasso estimator
\begin{align}\label{eq:node_wise_lasso}
&\widehat\phi^\ell \in \arg\min_{\phi \in \mathbb{R}^{k_0-1}} \norm{Z^\ell-Z^{-\ell} \phi}^2 + \mu_{\text{node}}^\ell \norm{\frac{1}{\sqrt{NM}} S_{-\ell} \widehat\Upsilon_{\text{node}}^\ell\phi}_1
\end{align}
of the $\ell$-th column $Z^{\ell}$ on all the other $(p-1)$ columns $Z^{-\ell}$ for each $\ell \in \{1,...,p\}$, where $\mu_{\text{node}}^\ell \in [0,\infty)$ is a regularization tuning parameter, $\widehat\Upsilon_{\text{node}}^\ell$ is some diagonal normalization matrix for each $\ell \in \{1,...,p\}$, and $S_{-\ell}$ is the $(p-1) \times (p-1)$ matrix obtained by removing the $\ell$-th row and the $\ell$-th column of 
\begin{align*}
S &=
\left[\begin{array}{ccc}
\sqrt{NM}I_{k_0} & 0 & 0\\
0 & \sqrt{M}I_{N_0} & 0\\
0 & 0 & \sqrt{N}I_{M_0}
\end{array}\right].
\end{align*}
In practice, the regularization tuning parameter $ \mu_{\text{node}}^\ell$ can be chosen using a cross validation via software packages.

Once the nodewise lasso estimates $\widehat\phi^\ell$ are obtained, a $p \times p$ matrix $\widehat\Theta$ approximating the inverse Gram matrix can be constructed by
\begin{align}
\widehat\Theta &=
\left[\begin{array}{ccccc}
\widehat\tau_1^{-2} & 0 & \cdots & 0 & 0\\
0 & \widehat\tau_2^{-2} & & 0 & 0\\
\vdots & & \ddots & & \vdots\\
0 & 0 & & \widehat\tau_{p-1}^{-2} & 0\\
0 & 0 & \cdots & 0 & \widehat\tau_{p}^{-2}
\end{array}\right]
\left[\begin{array}{ccccc}
1 & -\widehat\phi_1^1 & \cdots & -\widehat\phi_{p-2}^1 & -\widehat\phi_{p-1}^1\\
-\widehat\phi_1^2 & 1 & \cdots & -\widehat\phi_{p-2}^2 & -\widehat\phi_{p-1}^2\\
\vdots & & \ddots & & \vdots\\
-\widehat\phi_1^{p-1} & -\widehat\phi_2^{p-1} & \cdots & 1 & -\widehat\phi_{p-1}^{p-1}\\
-\widehat\phi_1^{p} & -\widehat\phi_2^{p} & \cdots & -\widehat\phi_{p-1}^{p} & 1
\end{array}\right],
\label{eq:theta_hat}
\end{align}
with $\widehat\tau_\ell$ given by
\begin{align*}
&\widehat\tau_\ell^2 = \frac{1}{NM} \norm{Z_\ell - Z_{-\ell}\widehat\phi^\ell}^2 + \frac{\mu^\ell_{\text{node}}}{NM} \norm{\frac{1}{\sqrt{NM}} S_{-\ell} \widehat\Upsilon_{-\ell} \widehat\phi^\ell}_1
\end{align*}
for each $\ell \in \{1,...,p\}$
and $\widehat\phi^\ell_l$ denoting the $l$-th coordinate of the nodewise lasso estimate $\widehat\phi^\ell$ for each $\ell \in \{1,...,p\}$ and $l \in \{1,...,p-1\}$.
\bigskip\\
{\bf Step 3:}
The shrinkage by the regularization $\mu P(\boldsymbol{\eta})$ forces a sub-vector of the lasso estimates $\widehat{\boldsymbol{\eta}}$ to be zero, and this mechanism serves as means of model selection.
Since this regularization biases the second-stage lasso estimator $\widehat{\boldsymbol{\eta}}$, we further `de-bias' it according to
\begin{align}
\label{eq:de-biased}
\widetilde {\boldsymbol{\eta}}_\ell = \widehat {\boldsymbol{\eta}}_\ell +\frac{1}{NM}\widehat \Theta_\ell' Z'(Y-Z\widehat {\boldsymbol{\eta}}),
\end{align}
for each $\ell \in [p]$,
where 
$\widehat\Theta_\ell$ is the $\ell$-th column of $\widehat\Theta$
and
$\widehat\Theta$ is the $p \times p$ approximate inverse Gram matrix constructed in Step 2.
The sub-vectors of $\widetilde{\boldsymbol{\eta}}$ will be denoted by $\widetilde{\boldsymbol{\eta}} = \left( \widetilde{\boldsymbol{\beta}}', \widetilde{\boldsymbol{\alpha}}', \widetilde{\boldsymbol{\gamma}} \right)'$.
\bigskip\\
{\bf Step 4:}
The asymptotic variance of $\sqrt{NM}\left(\widetilde{\boldsymbol{\beta}}_\ell - \boldsymbol{\beta}_\ell\right)$ for $\ell \in \{1,...,k_0\}$ is approximated by
\begin{align*}
\widehat V_{\ell\ell} =  \widehat\Theta_\ell' \widehat\Omega \widehat\Theta_\ell
\end{align*}
where 
$\widehat\Theta_\ell$ is defined in Step 3,
\begin{align*}
\widehat\Omega = \frac{1}{NM} \sum_{i=1}^N\sum_{j=1}^M \left(\sum_{t=1}^TZ_{ijt}\widehat\varepsilon_{ijt}\right) \left(\sum_{t=1}^TZ_{ijt}\widehat\varepsilon_{ijt}\right)',
\end{align*}
and $\widehat\varepsilon_{ijt}$ is the residual from the lasso in Step 1.

\section{The Main Theory}
\label{sec:theory}

Define the de-biased lasso estimator by
\begin{align}\label{eq:de_biased}
\widetilde{\boldsymbol{\eta}} = \widehat{\boldsymbol{\eta}} - \frac{\mu}{NM} \widehat \Theta P'(\widehat{\boldsymbol{\eta}}),
\end{align}
where $P'$ denotes the sub-gradient of $P$.
Recall that the sub-vectors of $\widetilde{\boldsymbol{\eta}}$ are denoted by $\widetilde{\boldsymbol{\eta}} = [ \widetilde{\boldsymbol{\beta}}', \widetilde{\boldsymbol{\alpha}}', \widetilde{\boldsymbol{\gamma}} ]'$, corresponding to $\overline{\boldsymbol{\eta}} = [\boldsymbol{\beta}' \ \overline{\boldsymbol{\alpha}}' \ \overline{\boldsymbol{\gamma}}']'$.
This section presents a general limit distribution result for each coordinate of the de-biased lasso estimator $\widetilde{\boldsymbol{\beta}}$ for the coefficients of $x_{ijt}$.
We focus on short panels with fixed $T$ and large $(N,M)$, although an extension to large $T$ cases may be feasible with alternative assumptions.
While we maintain high-level assumptions in the current section for the sake of generality, we will follow up with lower-level sufficient conditions in Section \ref{sec:sufficient_conditions}.
Define the $p \times p$ rate-adjusted Gram matrix
\begin{align}\label{eq:psi_bar}
\bar\Psi=
\begin{bmatrix}
\frac{1}{NM}X'X & \frac{1}{M\sqrt{N}} X'D_1 & \frac{1}{N\sqrt{M}} X'D_2\\
 \frac{1}{M\sqrt{N}} D_1'X & \frac{1}{M}D_1'D_1 & \frac{1}{\sqrt{NM}}
 D_1'D_2\\
\frac{1}{N\sqrt{M}} D_2'X& \frac{1}{\sqrt{NM}} D_2'D_1& \frac{1}{N} D_2' D_2 
\end{bmatrix},
\end{align}
Let $[n] = \{1,...,n\}$ for any $n \in \mathbb{N}$.
With these notations, consider the following assumption.

\begin{assumption}[Asymptotic Normality]\label{a:asymptotic_normality}
For all $(N,M)$, there exists a column random vector $\widehat \Theta_l$ such that the following conditions hold for an $(N,M)$-dependent choice of $\mu$ as $N,M\rightarrow\infty$.
\begin{enumerate}[(i)]
\item $\max_{l\in[k_0]} \abs{ \sqrt{NM}(\widehat\Theta_l' Q\bar\Psi Q  -e_l')(\widehat {\boldsymbol{\eta}}  - \overline{\boldsymbol{\eta}} ) } = o_p(1)$.
\item $\max_{l\in[k_0]} \abs{ \widehat \Theta_l' Z'R/\sqrt{NM} }=o_p(1)$.
\item For each $l\in[k_0]$, there exists $V_{ll} \in (0,\infty)$ that can depend on $(N,M)$ such that
\begin{align*}
V_{ll}^{-1/2}\widehat\Theta_l' Z'\varepsilon /\sqrt{NM} \leadsto N(0,1).
\end{align*}
\end{enumerate}
\end{assumption}

In the current general theoretical discussions, Assumption \ref{a:asymptotic_normality} merely requires an existence of some $\widehat \Theta_l$ satisfying the three conditions, and does not say how it should be constructed.
Recall that the overview of the method in Section \ref{sec:method} suggests a concrete way to construct such $\widehat \Theta_l$.
Section \ref{sec:sufficient_conditions_and_asymptotic_variance} ahead will discuss lower-level sufficient conditions to guarantee that such a concrete construction of $\widehat \Theta_l$ satisfies the three high-level conditions in Assumption \ref{a:asymptotic_normality}.

\begin{theorem}[Asymptotic Normality]\label{theorem:asymptotic_normality}
Suppose that Assumption \ref{a:asymptotic_normality} (i)--(ii) are satisfied. Then,
\begin{align*}
\widetilde {\boldsymbol{\eta}}_l - \overline{\boldsymbol{\eta}}_l = \frac{1}{NM} \widehat \Theta'_l Z'\varepsilon+o_p \left( 1/\sqrt{NM} \right)
\end{align*}
for each $l \in [p]$.
Furthermore, if Assumption \ref{a:asymptotic_normality} (iii) is satisfied in addition,  then we have
\begin{align*}
\sqrt{NM}(\widetilde {\boldsymbol{\beta}}_l - {\boldsymbol{\beta}}_l)\leadsto  N(0,V_{ll})
\end{align*}
for each $l\in [k_0]$.
\end{theorem}

A proof is found in Appendix \ref{sec:proof:theorem:asymptotic_normality}.

\begin{remark}
The de-biased lasso estimator $\widetilde {\boldsymbol{\eta}}_l = \widehat {\boldsymbol{\eta}}_l -\frac{\mu }{NM}\widehat \Theta'_l P'(\widehat {\boldsymbol{\eta}})$ can be also rewritten by replacing $\mu P'(\widehat {\boldsymbol{\eta}})$ by $-Z'(Y-Z\widehat {\boldsymbol{\eta}})$ following the K.K.T. condition, i.e., $\widetilde {\boldsymbol{\eta}}_l = \widehat {\boldsymbol{\eta}}_l +\frac{1}{NM}\widehat \Theta_l' Z'(Y-Z\widehat {\boldsymbol{\eta}})$. 
This representation yields the de-biased lasso formula proposed in (\ref{eq:de-biased}).
\end{remark}

\section{Sufficient Conditions and Variance Estimation}
\label{sec:sufficient_conditions_and_asymptotic_variance}

In this section, we propose lower-level sufficient conditions for the high-level general statements in Assumption \ref{a:asymptotic_normality}.
These conditions provide a theoretical guarantee for the concrete practical procedure of Section \ref{sec:method} to work.
While the general limit distribution result in Theorem \ref{theorem:asymptotic_normality} did not specify a concrete form of the asymptotic variance $V_{ll}$, the current section also provides a formula for it under these sufficient conditions.
Furthermore, we propose an analog variance estimator $\widehat V_{ll}$, and show its consistency under these sufficient conditions.

Throughout this section, we will assume $\widehat \Upsilon=I_p$ and $\widehat \Upsilon_{\text{node},l}=I_{p-1}$ for all $l\in[k_0]$ for simplicity, although these restrictions are not essential at all.
We use the following notations for the parameter supports:
$J_1 = \text{supp}(\boldsymbol{\beta}),$
$J_2 = \text{supp}(\overline{\boldsymbol{\alpha}}),$
$J_3 = \text{supp}(\overline{\boldsymbol{\gamma}}),$
and
$J = \text{supp}(\overline{\boldsymbol{\eta}})$,
Their cardinalities are denoted by 
$s_1 = |J_1|,$
$s_2 = |J_2|,$
$s_3 = |J_3|,$
and
$s = |J|.$
We note that $s$ is non-decreasing in $N$ and/or $M$.
Similarly to the decomposition (\ref{eq:representation}) for the main regression model, we also consider the decomposition
\begin{align}
&Z_l= Z_{-l}\phi^l +r_l + \zeta_l,
\label{eq:auxiliary_decomposition}
\\
&E[Z_{-l}\zeta_l]=0
\nonumber
\end{align}
for each coordinate $l\in [k_0]$ of the regressors.

\subsection{Sufficient Conditions}
\label{sec:sufficient_conditions}

We present sufficient conditions as five modules, Assumptions \ref{a:sparsity}, \ref{a:moments}, \ref{a:sparse_eigenvalues}, \ref{a:Theta}, and \ref{a:variance}, listed below.

\begin{assumption}[Approximate Sparsity]\label{a:sparsity}
(1) $\| \overline{\boldsymbol{\eta}} \|\le K$. 
(2) $\|R\|\le c_s \lesssim \sqrt{s }$ with probability $1-o(1)$.
(3) 
$
\| Z'R \| = o_p\Big( \sqrt{NM} \Big).
$
\end{assumption}

Recall that the fixed effects $\boldsymbol{\alpha}$ are decomposed into $\overline{\boldsymbol{\alpha}}$ and $\boldsymbol{\alpha} - \overline{\boldsymbol{\alpha}}$ such that (\ref{eq:alpha_decomposition}) is satisfied, and the fixed effects $\boldsymbol{\gamma}$ are decomposed into $\overline{\boldsymbol{\gamma}}$ and $\boldsymbol{\gamma} - \overline{\boldsymbol{\gamma}}$ such that (\ref{eq:gamma_decomposition}) is satisfied.
These conditions (\ref{eq:alpha_decomposition}) and (\ref{eq:gamma_decomposition}) are imposed to satisfy Assumption \ref{a:sparsity} (1) and (2).
Assumption \ref{a:sparsity} (3) can be relaxed to a weaker condition,\footnote{For example, $\sup_{\substack{\|\xi\|=1 \\ \|\xi\|_0 = Cs}} \| \xi' Z'R \| = o_p(\sqrt{NM})$ for some finite positive $C$.} but we present the current condition for its better interpretation.

\begin{remark}[Discussion of the Approximate Sparsity Condition]\label{remark:approximate_sparsity}
We emphasize that the approximate sparsity condition of Assumption \ref{a:sparsity}  (together with Assumption \ref{a:Theta} (4) to be stated below) allows for many and even all the fixed effects (i.e., $\boldsymbol{\eta}$ as opposed to $\overline{\boldsymbol{\eta}}$) to be nonzero.
The assumption should be interpreted as a requirement for how the fixed effects can be decomposed into the sparse components ($\overline{\boldsymbol{\alpha}}$ and $\overline{\boldsymbol{\gamma}}$) and the remaining components ($\boldsymbol{\alpha} - \overline{\boldsymbol{\alpha}}$ and $\boldsymbol{\gamma} - \overline{\boldsymbol{\gamma}}$) generating
$
R
= D_1 \left(\boldsymbol{\alpha} - \overline{\boldsymbol{\alpha}}\right)
+ D_2 \left(\boldsymbol{\gamma} - \overline{\boldsymbol{\gamma}}\right).
$
Indeed, the assumption implicitly imposes a non-trivial restriction on sampling procedures.
For example, an i.i.d. sampling of fixed effects is not accommodated, although this feature does not contradict with our sampling assumption to be stated below as Assumption \ref{a:moments}.
With this said, the same limitations apply to all the preceding papers (cf. Section \ref{sec:introduction}) that employ (approximate) sparsity conditions on fixed effects in panel data.
In fact, the approximate sparsity is a rather plausible assumption for the sampling process in the context of our motivating application (\ref{eq:matyas}).
In gravity analysis of trade, researchers initially used only the G7 countries, later added the OECD countries, and smaller economies have been added more recently.
Nearly half of all import and export flows are determined by the top ten largest economies.
Newly added countries to the sample tend to have very small trade volumes.
This sampling process entails fixed effects taking smaller values as sample size increases, and it does not contradict with the approximate sparsity requirement.
Section \ref{sec:approximate_sparsity_in_gravity_analysis} elaborates on the approximate sparsity in trade volumes based on the actual world trade data.
$\triangle$
\end{remark}

\begin{assumption}[Moments]\label{a:moments}
For each $(N,M)$, the random vectors $(Y'_{ij1},Z'_{ij1},...,Y'_{ijT},Z'_{ijT})'$, $(i,j)\in [N]\times [M]$, are independently distributed. 
Furthermore, there exist $q \in (4, \infty)$ and $K \in (0,\infty)$ not depending on $(N,M)$ such that the following conditions hold for all $l\in[k_0]$.
\begin{enumerate}[(1)]
\item
$
\Big( \frac{1}{NM}\sumi\sumj E\Big[\max_{ t\le T}\|X_{ijt}\|^{2q}_\infty\Big]\Big)^{1/2q}\le B_{NM}
$
and
$
\Big( E|X_{ijt,l}|^{2q} \Big)^{1/2q} \le K
$
hold for all $i,j,t,l$, where $B_{MN}$ satisfies $B_{NM}\sqrt{\log (p \vee (NM))}\lesssim (NM)^{1/2-1/q}$;
\item $\|(D_1,D_2)\|_\infty=1$; and
\item $\frac{1}{NM}\sumi\sumj\sumt E\varepsilon_{ijt}^{2q}\vee \frac{1}{NM}\sumi\sumj\sumt E(\zeta^l_{ijt})^{2q}\le K^{2q} <\infty$.
\end{enumerate}
\end{assumption}

For any squared matrix $A$, define the sparse eigenvalues by
\begin{align*}
\varphi_{\min} (A,m)=\underset{ \substack{
\|\xi\|=1\\ \|\xi\|_0\le m
}}{\inf}\xi'A \xi 
\:\text{ and }\:
 \varphi_{\max} (A,m) =
 \underset{ \substack{
\|\xi\|=1\\ \|\xi\|_0\le m
}}{\sup}\xi'A \xi .
\end{align*} 
With these notations, we state the following assumption of sparse eigenvalues for the rate-adjusted Gram matrix $\bar\Psi$ defined in (\ref{eq:psi_bar}).

\begin{assumption}[Sparse Eigenvalues]\label{a:sparse_eigenvalues}
For any $C > 0$, there exist constants $0<\underline k < \overline k<\infty$, not depending on $(N,M)$, such that 
\begin{align*}
\underline k \le \varphi_{\min} (\bar\Psi,Cs)\le\varphi_{\max} (\bar\Psi,Cs) \le \overline k
\end{align*}
with probability approaching one.
\end{assumption}

For each $(N,M)$, we write $\Psi = E\bar\Psi$ depending on $(N,M)$,
With this notation, the auxiliary decomposition (\ref{eq:auxiliary_decomposition}) is made according to the following conditions.

\begin{assumption}[Nuisance Parameters]\label{a:Theta}
The following conditions are satisfied.
\begin{enumerate}[(1)]
\item $\max_{l \in [k_0]}\|\phi^l\|_0\le s_l$ and $\max_{l \in [k_0]}\|\phi^l\|+(s_l)^{-1/2}\|\phi^l\|_1\le K$;
\item For all $l\in [k_0]$, $\|r_l\|\le \sqrt{s_l}$;
\item For all $(N,M)$, $0<L<\Lambda_{\min}(\Psi)<\Lambda_{\max}(\Psi)<U<\infty$ for $L$, $U$ independent of $(N,M)$;
\item $\max_{l\in[k_0]}(s_l\vee s)\sqrt{\frac{(\log (p \vee (NM)))^2}{N\wedge M}}=o(1)$.
\end{enumerate}
\end{assumption}

Accounting for the possible dependence, we define the cluster-robust variance matrix
\begin{align*}
\Omega=& E\left[ \frac{1}{NM}\sumi\sumj   \Big(
\sumt Z_{ijt} \varepsilon_{ijt}
\Big) \Big(
\sumt Z_{ijt} \varepsilon_{ijt}
\Big)'\right].
\end{align*}
For each $(N,M)$, we write $\Theta=(E[\frac{Z'Z}{NM}])^{-1}$ depending on $(N,M)$.
Let $\Theta_l$ denote the $l$-th column of $\Theta$.
We state the following assumption of finite and non-zero variance.

\begin{assumption}[Variance]\label{a:variance}
For any $(N,M)$ and for all $l\in[k_0]$,
$
\|\Omega\| <\infty
$
and
$
\Theta_l'\Omega\Theta_l \ge\underline k >0
$
for a constant $\underline k $ which is independent of the sample size.
\end{assumption}
\begin{remark}
Notice that the conditions above are imposed on the Gram matrices, $\bar \Psi$ and $\Psi$, re-weighted by effective sample size, rather than the original Gram matrices, $Z'Z/NM$ and $EZ'Z/NM$.
Assumption \ref{a:moments} is weaker than the common assumptions required in the literature, such as sub-gaussianity or uniform boundedness.
Assumption \ref{a:sparse_eigenvalues} is also assumed by \cite{BelloniChenChernozhukovHansen2012} and \cite{BelloniChernozhukovHansenKozbur2016}. It requires some small sub-matrices of the big $p\times p$ re-weighted
Gram matrix to be well-behaved. 
Lower level sufficient conditions are also possible by using Lemma P1 in \cite{BelloniChernozhukovChetverikovWei2018}, but are not pursued here. 
Assumption \ref{a:Theta} (1) and (2) impose sparsity on the nodewise regression parameters and the approximation errors. Assumption \ref{a:Theta} (3) requires $\Psi$, the expectation of the re-weighted Gram matrix, to be positive definite uniform over $(N,M)$. These are rather standard in the literature. Assumption \ref{a:Theta} limits the models that can be handled in terms of their dimensionality and sparsity. 
Note that we need only $ss_l(\log (p\vee (NM)))^2/(N \wedge M)=o(1)$, whereas an adaptation of the proof strategies of \citet{Kock2016} and \citet{KockTang2018} to our framework would entail $ss^2_l(\log(p\vee (NM)))^2/(N \wedge M)=o(1)$. Finally, Assumption \ref{a:variance} requires $\Omega$ in the sandwich form to be well-behaved.
\end{remark}

The following proposition states that Assumptions \ref{a:sparsity}, \ref{a:moments}, \ref{a:sparse_eigenvalues}, \ref{a:Theta}, and \ref{a:variance} are sufficient for the high-level conditions in Assumption \ref{a:asymptotic_normality}, with a concrete variance formula motivating the practical guideline of Section \ref{sec:method}.

\begin{proposition}\label{prop:sufficient}
Assumptions \ref{a:sparsity}, \ref{a:moments}, \ref{a:sparse_eigenvalues}, \ref{a:Theta}, and \ref{a:variance} imply Assumption \ref{a:asymptotic_normality} with $V_{ll} = \Theta_l' \Omega \Theta_l$.
\end{proposition}

A proof is found in Appendix \ref{sec:proof:prop:sufficient}.
Combining Theorem \ref{theorem:asymptotic_normality} and Proposition \ref{prop:sufficient} together, we state the following corollary. 

\begin{corollary}[Asymptotic Normality]\label{corollary:asymptotic_normality}
If Assumptions \ref{a:sparsity}, \ref{a:moments}, \ref{a:sparse_eigenvalues}, \ref{a:Theta}, and \ref{a:variance} are satisfied, then
\begin{align*}
\sqrt{NM}(\widetilde {\boldsymbol{\beta}}_l - {\boldsymbol{\beta}}_l)\leadsto  N(0,V_{ll})
\end{align*}
for each $l\in [k_0]$, where $V_{ll} = \Theta_l' \Omega \Theta_l$.
\end{corollary}
\begin{remark}
We conjecture that one can further enhance the results of Corollary \ref{corollary:asymptotic_normality} by showing the honesty property (uniform validity over a large set of parameters) of confidence intervals using the proposed procedure with no extra assumption by adapting the proof strategy of Theorem 3 of \citet{CanerKock2018} or Theorem 3 of \citet{KockTang2018} to our framework.
\end{remark}
\subsection{Asymptotic Variance Estimation}
\label{sec:asymptotic_variance}

Based on the asymptotic variance formula presented in Proposition \ref{prop:sufficient}, we suggest to compute the cluster-robust asymptotic variance of $\sqrt{NM}\left(\widetilde{\boldsymbol{\beta}}_\ell - \boldsymbol{\beta}_\ell\right)$ by
\begin{align*}
\widehat V_{ll} =&  \widehat\Theta_{l}' \widehat  \Omega \widehat\Theta_{l},
\end{align*}
as suggested in Section \ref{sec:method}.
This estimator is consistent in the current assumptions as formally stated in the following theorem.

\begin{theorem}[Variance Estimator]\label{theorem:variance_estimator}
If Assumptions \ref{a:sparsity}, \ref{a:moments}, \ref{a:sparse_eigenvalues}, \ref{a:Theta} and \ref{a:variance} are satisfied, then 
$$
\max_{l\in [k_0]}|\widehat V_{ll} - V_{ll}|=o_p(1).
$$
\end{theorem}

A proof is found in Appendix \ref{sec:proof:theorem:variance_estimator}.

\section{Approximate Sparsity in Gravity Analysis of Trade}
\label{sec:approximate_sparsity_in_gravity_analysis}

In this section, we discuss our key assumption, namely the assumption of approximate sparsity (Assumptions \ref{a:sparsity} and \ref{a:Theta} (4) -- also see Remark \ref{remark:approximate_sparsity}), in the gravity model (\ref{eq:matyas}) of international trade.
The idea behind the approximate sparsity assumption is that only a small number of observations have large fixed effect values, and the remaining majority of observations have relatively modest fixed effect values that can be summarized into the approximation error term $r_{ijt}$.
The assumption is likely satisfied in sampling processes where, after collecting observations with relatively large values of fixed effects (e.g., G7 and OECD countries), the remaining additions tend to have smaller values of fixed effects.
We argue that this is plausible in common settings such gravity analysis in international trade.

To illustratea this point, we retrieved data from the World Integrated Trade Solution (WITS) Database, a common source of trade flows and trade costs used in gravity analysis.\footnote{This database was developed by the World Bank in conjunction with the United Nations Conference on Trade and Development (UNCTAD), the International Trade Center, United Nations Statistical Division (UNSD) and the World Trade Organization (WTO). 
The database combines information on trade flows from the UN Comtrade database, tariff and non-tariff barriers from the UN TRAINS database, and the both preferential and MFN tariffs from the WTO's Integrated Data Base.}  We focus on country-specific import and export flows and aim to make two specific points.  
First, in any given year, trade is largely dominated by a few large countries.  
For instance, in 2015, the WITS database contains positive import flows for 237 countries and positive export flows for 232 countries.  
Nonetheless, nearly half of all import (respectively, export) flows are determined by the top 10 largest importers (respectively, exporters) alone.  
Not surprisingly, the largest importers are also the largest exporters.
Second, the importance of these countries has remained stable over time, despite the fact that (a) world trade has grown exponentially over time and (b) WITS records exports and imports for a substantially larger number of countries in recent years than it did even a few years ago.  
In this sense, the `new' additions to trade databases tend to have very small trade flows.

Using a country's share of world imports (Table \ref{tab:world_import_shares_over_time}) as a measure of `importer' importance or a country's share of world exports (Table \ref{tab:world_export_shares_over_time}) as a measure of exporter importance, we document the 10 largest trading nations every 5 years starting in 1990.  
We note the following three attributes of standard trade data:
(1) a small number of countries account for the large majority of world trade;
(2) whether a country represents a large or small fraction of trade flows changes slowly over time; and
(3) even though many developing countries have grown substantially since 1990, the average share of small countries has not changed very much.  This last feature is largely due to the fact that the `new' countries which are added to world trade databases are nearly always very small.
To make points (1) and (2) particularly clear, we would expect that a typical country would have an import/export share of roughly 0.5\% for a sample of about 200 countries. However, in any given year, fewer than 40 countries have import or export shares of 0.5\%.  
Of the countries which have been added to the import database since 1990, their average (median) import share was 0.07\% (0.01\%).  
Similarly, among the countries added to the export database since 1990, their average (median) export share was 0.09\% (0.03\%).  
Regardless of how to measure the size of these peripheral countries, their overall contribution to world trade is extremely small.

In summary, only a small number of observations have high trade volumes.
The large majority of remaining observations have very modest and almost negligible trade shares.
This pattern remains stable over time.
Since researchers first collect observations with large volumes (e.g., G7 and OECD countries), new additions to the data thereafter entail relatively small volumes.
This common sampling process in gravity analysis of international trade is compatible with our key assumption, namely the the assumption of approximate sparsity (Assumptions \ref{a:sparsity} and \ref{a:Theta} (4) -- also see Remark \ref{remark:approximate_sparsity}).

\section{Simulation Studies}
\label{sec:simulation_studies}

\subsection{Simulation Setting}

Consider the following three fixed effect models of three-dimensional panel data.
\begin{align*}
\text{Model (I):} & \qquad
y_{ijt} = {x}_{ijt} \beta 
+ \alpha_i + \gamma_j
+ \varepsilon_{ijt}
\\
\text{Model (II):} & \qquad
y_{ijt} = {x}_{ijt} \beta 
+ \alpha_i + \gamma_j + \lambda_t
+ \varepsilon_{ijt}
\\
\text{Model (III):} & \qquad
y_{ijt} = {x}_{ijt} \beta 
+ \alpha_{it} + \gamma_{jt}
+ \varepsilon_{ijt}
\end{align*}
Model (I) is nested by Model (II), and Model (II) is in turn nested by Model (III).
Therefore, Model (I) is the most parsimonious and subject to under-fitting, whereas Model (III) is the richest and subject to over-fitting.
If a researcher runs a fixed effect estimator under Model (I) when Model (II) or (III) is true, then the estimates generally suffers from mis-specification biases.
If a researcher runs a fixed effect estimator under Model (III) when Model (I) or (II) is true, then the estimates generally suffers from larger standard errors than necessary.

We run simulations for varying sizes of $N$ and $M=N-1$, while the length of time is set to $T=5$ throughout.
This setting follows from our asymptotic theory where $N$ and $M$ increases but $T$ does not.
The $i$ and $j$ fixed effects are generated by 
$\alpha_i \sim N\left(m_\alpha, s_\alpha^2 \left/\left(\sqrt{i} \cdot (\log (i+1))^3\right)\right.\right)$
and
$\gamma_j \sim N\left(m_\gamma, s_\gamma^2 \left/\left(\sqrt{j} \cdot (\log (j+1))^3\right)\right.\right)$
independently, where 
$m_\alpha=m_\gamma=0$ and 
$s_\alpha=s_\gamma=1$.
The $t$ fixed effects are generated by $\lambda_t=0$ for all $t$ but for one year $t$ when a universal shock of $\lambda_t = 2$ is applied.
The $it$ and $jt$ fixed effects are generated by
$\alpha_{it} \sim N\left(m_{\alpha}, s_{\alpha}^2 \left/\left(\sqrt{i} \cdot (\log (i+1))^3\right)\right.\right)$,
$\gamma_{jt} \sim N\left(m_{\gamma}, s_{\gamma}^2 \left/\left(\sqrt{j} \cdot (\log (j+1))^3\right)\right.\right)$,
$m_{\alpha}=m_{\gamma}=0$, and
$s_{\alpha}=s_{\gamma}=1$.
We generate $X$ dependently on the fixed effects according to the mixture
\begin{align*}
x_{ijt} = m_x + s_x \cdot \left[ (1-\rho) \cdot \tilde x_{ijt} + \rho F_{ijt} \right],
\end{align*}
where $m_x=0$, $s_x=2$, $\rho=0.5$, $\tilde x_{ijt} \sim N(0,1)$, 
and $F_{ijt}$ is the standardized sum of fixed effects for the unit $(i,j,t)$, i.e.,
\begin{align*}
\text{Under Model (I):}& \qquad
F_{ijt} = \left(\alpha_i + \gamma_j\right) \left/ \sqrt{\frac{1}{NM}\sum_{i=1}^N\sum_{j=1}^M \left(\alpha_i + \gamma_j\right)^2}\right.
\\
\text{Under Model (II):}& \qquad
F_{ijt} = \left(\alpha_i + \gamma_j + \lambda_t\right) \left/ \sqrt{\frac{1}{NMT}\sum_{i=1}^N\sum_{j=1}^M\sum_{t=1}^T \left(\alpha_i + \gamma_j + \lambda_t\right)^2}\right.
\\
\text{Under Model (II):}& \qquad
F_{ijt} = \left(\alpha_{it} + \gamma_{jt}\right) \left/ \sqrt{\frac{1}{NMT}\sum_{i=1}^N\sum_{j=1}^M\sum_{t=1}^T \left(\alpha_{it} + \gamma_{jt}\right)^2}\right.
\end{align*}
for each $(i,j,t) \in \{1,...,N\} \times \{1,...,M\} \times \{1,...,T\}$.
The error term is generated by $\varepsilon_{ijt} \sim N(m_\varepsilon,s_\varepsilon^2)$ independently where $m_\varepsilon=0$ and $s_\varepsilon=10$.
The main coefficient of interest is set to $\beta = 1$.
Each set of simulations consists of 10,000 Monte Carlo iterations of data generation, estimation, and inference.

We compare five methods of estimation and inference.
These are the OLS without any individual fixed effects, the fixed effect estimator based on Model (I), the fixed effect estimator based on Model (II), the fixed effect estimator based on Model (III), and our proposed de-biased lasso estimator and post-selection inference.
Note that the OLS is always under-fitting the true data generating model, and hence is expected to produce mis-specification biases.
The fixed effect estimator based on Model (I) is correctly specified when the true data generating model is Model (I), but is under-fitting Model (II) and Model (III).
The fixed effect estimator based on Model (II) is over-fitting Model (I), correctly specified when the true data generating model is Model (II), and under-fitting Model (III).
The fixed effect estimator based on Model (III) is over-fitting Model (I) and Model (II), but is correctly specified when the true data generating model is Model (III).

\subsection{Simulation Results}

Table \ref{tab:simulation} displays Monte Carlo simulation results under Model (I) (top panel), Model (II) (middle panel), and Model (III) (bottom panel) with the sample size $N=10$ ($NMT=450$).
Similarly, Tables \ref{tab:simulation15} and \ref{tab:simulation20} display Monte Carlo simulation results with the sample sizes $N=15$ ($NMT=1050$) and $N=200$ ($NMT=1900$), respectively.
The displayed statistics are the averages, biases, standard deviations, and root mean squared errors of estimates.
Also displayed are the coverage frequencies of the true value of $\beta$ by the 95\% confidence intervals.
The first column of each table shows the OLS results without any individual fixed effects.
The next three columns of each table show results of fixed effect estimators based on estimating equations of Model (I), Model (II), and Model (III).
We shall call them FE-I, FE-II, and FE-III for succinctness.
The last column of each table shows results of our proposed de-biased lasso estimator with valid post-selection inference.
We shall call it POST for succinctness.

In the top panel of each table, where the true data generating model is Model (I), OLS is biased while FE-I, FE-II, and FE-III yield little biases.
These results are consistent with the current simulation setting as OLS mis-specifies the true model while FE-I, FE-II, and FE-III correctly specify the true model.
The bias of POST is in the middle between that of OLS and those of FE-I, FE-II, and FE-III.
In other words, POST is de-biased to some extent but not to the full extent so that desired balances between the bias and variance are maintained.
OLS yields a smaller standard deviation than FE-I or FE-II, and FE-III yields by far the largest standard deviation.
These results are also consistent with the fact that OLS is the most parsimonious while FE-III is the most redundant in specification.
POST yields an even smaller standard deviation than OLS.
FE-I, as the oracle estimator, yields a smaller root mean square error than OLS, FE-II, or FE-III. 
Furthermore, POST yields an even smaller root mean square error than the oracle estimator, FE-I.
The coverage frequency by FE-I, as the oracle estimator, is closer to the nominal level 95\% than those of OLS, FE-II, or FE-III.
Furthermore, POST yields the coverage frequency as close to the nominal level as the oracle estimator, FE-I.
In summary, we observe that, when the true model is parsimonious, POST is more efficient than redundantly rich models and allows for as accurate inference as the oracle estimator.

In the middle panel of each table, where the true data generating model is Model (II), OLS and FE-I are biased while FE-II and FE-III yield little biases.
These results are consistent with the current simulation setting as OLS and FE-I mis-specify the true model while FE-II and FE-III correctly specify the true model.
The bias of POST is slightly larger than those of FE-II and FE-III, but much smaller than those of OLS and FE-I.
In other words, POST is de-biased to a large extent but not to the full extent so that desired balances between the bias and variance are maintained.
FE-II, as the oracle estimator, yields a smaller root mean square error than OLS, FE-I, or FE-III. 
Furthermore, POST yields an even smaller root mean square error than the oracle estimator, FE-II.
The coverage frequency by FE-II, as the oracle estimator, is closer to the nominal level 95\% than those of OLS, FE-I, or FE-III.
POST yields the coverage frequency as close to the nominal level as the oracle estimator, FE-II.
In summary, we observe that POST is more precise than biased parsimonious estimators, is more efficient than redundant estimators, and allows for as accurate inference as the oracle estimator.

In the bottom panel of each table, where the true data generating model is Model (III), OLS, FE-I, and FE-II are biased while FE-III yields a little bias.
These results are consistent with the current simulation setting as OLS, FE-I, and FE-II mis-specify the true model while FE-III correctly specifies the true model.
The bias of POST is in the middle between those of OLS, FE-I, and FE-II and that of FE-III.
In other words, POST is de-biased to some extent but not to the full extent so that desired balances between the bias and variance are maintained.
POST yields a smaller root mean square error than any other estimator, including the oracle estimator, FE-III.
POST also yields the coverage frequency closer to the nominal level than any estimator, including the oracle estimator, FE-III.
In summary, we observe that, when the true model is rich, POST is more precise than parsimonious estimators and allows for as accurate inference as the oracle estimator.

The simulation results reported above demonstrate that the proposed method (POST) can be used as a robustly applicable method of inference when a researcher does not know the correct fixed effect specification in practice.
We also implemented many additional sets of simulations under alternative data generating parameters, and confirm that the qualitative pattern of these additional results remain the same as those of our baseline setting presented above. 
Specifically, we consistently observe that POST is more precise than biased parsimonious estimators, is more efficient than redundant estimators, and allows for as accurate inference as the oracle estimator.


\section{Discussions}
\label{sec:discussions}

Three-dimensional panel models are widely used in empirical analysis of international trade, housing, migration, and consumer
price, among others.
Empirical researchers use various combinations of fixed effects for three-dimensional panels.
When a researcher imposes a parsimonious model and the true model is rich, then estimation based on the assumed parsimonious model generally incurs mis-specification biases.
When a researcher employs a rich model and the true model is parsimonious, then estimation based on the redundantly rich model generally incurs larger standard errors than necessary.
It is therefore useful for researchers to know correct models for an application of interest.
In this light, \citet{LuMiaoSu2018} propose methods of model selection in three-dimensional panel data.
In this paper, we advance this literature by proposing a method of post-selection inference for regression parameters.
We propose to use the lasso technique as means of model selection and to de-bias the lasso estimate, but our assumptions allow for many and even all fixed effects to be nonzero.
Simulation studies demonstrate that the proposed method is more precise than biased estimators by parsimonious models, is more efficient than noisy estimators by redundant models, and allows for as accurate inference as the oracle estimator.

We suggest a couple of directions for future research.
First, our model framework does not allow for $ij$ fixed effects, while $i$, $j$, $t$, $it$ and $jt$ fixed effects are allowed.
Although allowing for $ij$ fixed effects is not of interest in our motivating example,\footnote{In gravity models for international trade, the main parameters of interest are the coefficient of $DIST_{ij}$, interpreted as the trade elasticity or trade cost, and the coefficient of $TA_{ij}$, interpreted as the effects of bilateral trade agreements on trade volume. These parameters will not be identified once $ij$ fixed effects enter the model.} it may be possible to allow for such fixed effects provided that the asymptotic setting allows for large $T$ as well as large $N$ and/or large $M$.
Formal theoretical development for this case is left for future research.
Second, we conjecture that our limit distribution result can be extended to establish honest (uniformly valid) confidence intervals, and formal theoretical investigation of the honesty property is left for future research.

\newpage
\appendix
\section*{Mathematical Appendix}

Throughout, we use the following short-hand notations: 
$Q=S/\sqrt{NM}$ and $a = p \vee (NM)$. Also, for a matrix $A$, denote $\|A\|_\infty=\max_{i,j}|A_{i,j}|$.

\section{Proofs of the Main Results}
\label{sec:proofs_of_the_main_results}

\subsection{Proof of Theorem \ref{theorem:asymptotic_normality}}\label{sec:proof:theorem:asymptotic_normality}
\begin{proof}
The K.K.T. condition for the lasso program (\ref{eq:lasso}) gives
\begin{align*}
-Z'(Y-Z\widehat{\boldsymbol{\eta}}) + \mu P'(\widehat{\boldsymbol{\eta}})=0.
\end{align*}
Note that we have $S\bar\Psi S=  Z'Z$ by the definition of $\bar\Psi$, and thus
\begin{align*}
S\bar\Psi S(\widehat {\boldsymbol{\eta}}  - \overline{\boldsymbol{\eta}}  ) +\mu  P'(\widehat {\boldsymbol{\eta}}) = Z'\varepsilon + Z'R.
\end{align*}
Multiplying both sides by $\widehat \Theta_l'/\sqrt{NM}$, we have
\begin{align*}
\sqrt{NM}\widehat\Theta_l'
Q\bar\Psi Q  (\widehat {\boldsymbol{\eta}}  - \overline{\boldsymbol{\eta}} ) + \frac{\mu \widehat\Theta'_l P'(\widehat {\boldsymbol{\eta}})}{\sqrt{NM}} = \frac{\widehat\Theta_l'  Z'\varepsilon}{\sqrt{NM}} + \frac{\widehat\Theta'_l Z'R}{\sqrt{NM}},
\end{align*}
where $Q=S/\sqrt{NM}$.
Therefore, we have
\begin{align*}
\sqrt{NM}e_l'(\widehat {\boldsymbol{\eta}}-{\boldsymbol{\eta}})+\sqrt{NM}(\widehat\Theta_l' Q\bar\Psi Q  -e_l')(\widehat {\boldsymbol{\eta}}  - \overline{\boldsymbol{\eta}} ) + \frac{\mu \widehat\Theta'_l P'(\widehat {\boldsymbol{\eta}})}{\sqrt{NM}} = \frac{\widehat\Theta'_l  Z'\varepsilon}{\sqrt{NM}} + \frac{\widehat\Theta'_l Z'R}{\sqrt{NM}}.
\end{align*}
By Assumption \ref{a:asymptotic_normality} (i)--(ii) and the definition (\ref{eq:de_biased}) of the de-biased lasso, we obtain
\begin{align*}
\sqrt{NM}e_l'(\widetilde {\boldsymbol{\eta}}- \overline{\boldsymbol{\eta}}) = \frac{\widehat\Theta'_l  Z'\varepsilon}{\sqrt{NM}} + o_p(1).
\end{align*}
Applying Assumption \ref{a:asymptotic_normality} (iii) for each $l \in [k_0]$ yields the weak convergence result.
\end{proof}

\subsection{Proof of Proposition \ref{prop:sufficient}}\label{sec:proof:prop:sufficient}
\begin{proof}
The sufficiency of Assumptions \ref{a:sparsity}, \ref{a:moments}, \ref{a:sparse_eigenvalues}, and \ref{a:Theta} for Assumption \ref{a:asymptotic_normality} (i) is provided in Lemma \ref{lemma:approximation_error_Delta}.
The sufficiency of Assumptions \ref{a:sparsity}, \ref{a:moments}, \ref{a:sparse_eigenvalues}, and \ref{a:Theta} for Assumption \ref{a:asymptotic_normality} (ii) is provided in Lemma \ref{lemma:suff_2}.
The sufficiency of Assumptions \ref{a:moments}, \ref{a:sparse_eigenvalues}, \ref{a:Theta} and \ref{a:variance} for Assumption \ref{a:asymptotic_normality} (iii) is provided in Lemma \ref{lemma:suff_3}.
\end{proof}

\subsection{Proof of Theorem \ref{theorem:variance_estimator}}\label{sec:proof:theorem:variance_estimator}

\begin{proof}
We introduce the intermediate object defined by
$$
\tilde  \Omega=\sumi\sumj  S^{-1} \Big(
\sumt Z_{ijt} \varepsilon_{ijt}
\Big) \Big(
\sumt Z_{ijt} \varepsilon_{ijt}
\Big)' S^{-1}.
$$
Lemma \ref{lemma:empirical_pre-sparsity} under Assumptions \ref{a:sparsity}, \ref{a:moments}, \ref{a:sparse_eigenvalues} and \ref{a:Theta} yields $\max_{l\in[p]}\|\hat \Theta_l\|_0\le Cs_l$ with probability $1-o(1)$ for some $C$ large enough for all $l\in [k_0]$.
Therefore, we obtain the decomposition
\begin{align}
&|\hat\Theta_l' \hat\Omega \hat\Theta_l-\Theta_l'\Omega \Theta_l| \nonumber\\
\le&
 |\hat\Theta_{l}' \hat\Omega \hat\Theta_{l}-\hat\Theta_{l}' \Omega \hat\Theta_{l}| +|\hat\Theta_{l}' \Omega \hat\Theta_{l}-\Theta_{l}'\Omega \Theta_{l}|\nonumber\\
 \le 
 &\|\hat \Theta_l\|^2 \max_{ \substack{\|\xi\|=1\\
 \|\xi\|_0\le Cs_l}}\xi'(\hat\Omega - \tilde\Omega)\xi  
 +
  \|\hat \Theta_l\|^2_1 \|\tilde\Omega - \Omega \|_\infty
 +
\|\hat \Theta_{l} - \Theta_{l}\|^2\max_{ \substack{\|\xi\|=1\\
 \|\xi\|_0\le Cs_l}} \xi'\Omega\xi
 +
 2\|\Omega \Theta_{l}\| \|\hat \Theta_{l} - \Theta_{l}\|\label{eq:variance_est}
\end{align}
for all $l\in [k_0]$
By Lemma \ref{lemma:Theta_1_way} under Under Assumptions \ref{a:moments}, \ref{a:sparse_eigenvalues}, and \ref{a:Theta} , it suffices to bound 
$
\max_{ \substack{\|\xi\|=1\\
 \|\xi\|_0\le Cs_l}}\xi'(\hat\Omega - \tilde\Omega)\xi  
$ 
and
$
\|\hat \Theta_l\|^2_1 \|\tilde\Omega - \Omega \|_\infty
$
on the right-hand side.

We first bound 
$
\|\hat \Theta_l\|^2_1 \|\tilde\Omega - \Omega \|_\infty
$
on the right-hand side of (\ref{eq:variance_est}). 
Since $\max_{l\in[p]}\|\hat \Theta_l\|_0 = O_p( s_l)$ with probability approaching one and $\max_{l\in[p]}\|\hat\Theta_l\|=O_p(1)$, we have $\|\hat \Theta_l\|_1=O_p( \sqrt{s_l})$ uniformly over $l\in[k_0]$. 
By an application of Lemma \ref{lemma:concentration_inequality_CCK}, we have
\begin{align*}
\|\tilde\Omega - \Omega\|_\infty
 \le & T
  \max_{t\in [T]}\max_{l\in[p]}\Big|\frac{1}{NM}\sumi\sumj (Z_{ijt,l}^2\varepsilon_{ijt}^2-E[Z_{ijt,l}^2\varepsilon_{ijt}^2])\Big|\\
  \lesssim & \sqrt{\frac{\sigma^2 \log a}{NM}} + \frac{B\log a}{NM}
\end{align*}
with probability at least $1-o(1)$,
where
\begin{align*}
\sigma^2&=\max_{t\in[T],l\in [p]}\frac{1}{NM}\sumi\sumj E[Z_{ijt,l}^2\varepsilon_{ijt}^2]
\\
&\le \max_{t\in[T],l\in [p]} \sqrt{\frac{1}{NM}\sumi\sumj E[Z_{ijt,l}^2] } \sqrt{\frac{1}{NM}\sumi\sumj E[\varepsilon_{ijt}^2] }=O(1) 
\end{align*}
under Assumption \ref{a:moments}, and
\begin{align*}
B^2=& E[\max_{i,j,t} \|Z_{ijt}\varepsilon_{ijt}\|^2_\infty]\\
\le& (E[\max_{i,j,t} \|Z_{ijt}\varepsilon_{ijt}\|^q_\infty])^{2/q}\\
\le& (NM)^{2/q}(\frac{1}{NM}\sumi\sumj\sumt E[\|Z_{ijt}\varepsilon_{ijt}\|^q_\infty])^{2/q}\\
\le& (NM)^{2/q}\Big\{\Big(\frac{1}{NM}\sumi\sumj\sumt E[\|Z_{ijt}\|^{2q}_\infty]\Big)^{1/2}
\Big(\frac{1}{NM}\sumi\sumj\sumt E[\varepsilon_{ijt}^{2q}]\Big)^{1/2}
\Big\}^{2/q}\\
\lesssim &  (NM)^{2/q} B_{NM}^2 O(1)
\end{align*}
under Assumption \ref{a:moments}.
Therefore, we obtain
\begin{align*}
 \sqrt{\frac{\sigma^2 \log a}{NM}} + \frac{B\log a}{NM}\lesssim \sqrt{\frac{\log a}{NM}} + \frac{B_{NM}\log a}{(NM)^{1-1/q}} = O\Big(\sqrt{\frac{\log a}{NM}}\Big)
\end{align*}
where the last rate follows from Assumption \ref{a:moments} (i).
Combining these results, we obtain
\begin{align*}
 \|\hat \Theta_l\|^2_1 \|\tilde\Omega - \Omega \|_\infty=  O_p\Big(\sqrt{\frac{s_l^2\log a}{NM}}\Big).
\end{align*}

We next bound 
$
\max_{ \substack{\|\xi\|=1\\
 \|\xi\|_0\le Cs_l}}\xi'(\hat\Omega - \tilde\Omega)\xi  
$ 
on the right-hand side of (\ref{eq:variance_est}). 
Note that $\hat \varepsilon =\varepsilon+R- Z(\widehat{\boldsymbol{\eta}} - \overline{\boldsymbol{\eta}})$.
Thus,
\begin{align*}
&\max_{ \substack{\|\xi\|=1\\
 \|\xi\|_0\le Cs_l}}\xi'(\hat \Omega-\tilde \Omega )\xi\\
 =
&\max_{ \substack{\|\xi\|=1\\
 \|\xi\|_0\le Cs_l}}\xi'\sumi\sumj \frac{1}{NM}\Big\{
 \Big(\sumt Z_{ijt}\varepsilon_{ijt}\Big)\Big(\sumt  Z_{ijt}r_{ijt} \Big)'
 - \Big(\sumt Z_{ijt}\varepsilon_{ijt}\Big)\Big(\sumt  Z_{ijt}Z'_{ijt}(\widehat{\boldsymbol{\eta}} - \overline{\boldsymbol{\eta}}) \Big)'\\
 &+ \Big(\sumt Z_{ijt}R_{ijt}\Big)\Big(\sumt  Z_{ijt}\varepsilon_{ijt} \Big)'
  + \Big(\sumt Z_{ijt}R_{ijt}\Big)\Big( \sumt Z_{ijt}R_{ijt} \Big)'\\
& - \Big(\sumt Z_{ijt}R_{ijt}\Big)\Big(\sumt  Z_{ijt}Z'_{ijt}(\widehat{\boldsymbol{\eta}} -\overline{\boldsymbol{\eta}}) \Big)'
 - \Big(\sumt Z_{ijt}Z'_{ijt}(\widehat{\boldsymbol{\eta}} - \overline{\boldsymbol{\eta}})\Big)\Big(\sumt  Z_{ijt}\varepsilon_{ijt} \Big)'\\
& - \Big(\sumt Z_{ijt}Z'_{ijt}(\widehat{\boldsymbol{\eta}} - \overline{\boldsymbol{\eta}})\Big)\Big(\sumt  Z_{ijt} R_{ijt}\Big)'
 + \Big(\sumt Z_{ijt}X'_{ijt}(\widehat{\boldsymbol{\eta}} - \overline{\boldsymbol{\eta}})\Big)\Big(\sumt  Z_{ijt}Z'_{ijt}(\widehat{\boldsymbol{\eta}} - \overline{\boldsymbol{\eta}}) \Big)'\Big\}\xi\\
 \le & (1)+(2)+(3)+(4)+(5)+(6)+(7)+(8).
\end{align*}
We bound each term of the last eight terms separately. 
First, Cauchy-Schwartz's inequality yields
\begin{align*}
(8)=&
\max_{ \substack{\|\xi\|=1\\
 \|\xi\|_0\le Cs_l}}
 \xi'\sumi\sumj \frac{1}{NM}\Big\{ \Big(\sumt Z_{ijt}Z'_{ijt}(\widehat{\boldsymbol{\eta}} - \overline{\boldsymbol{\eta}})\Big)\Big(\sumt  Z_{ijt}Z'_{ijt}(\widehat{\boldsymbol{\eta}} - \overline{\boldsymbol{\eta}}) \Big)'\Big\}\xi\\
 \le&
 T^2\max_{t\in [T]}\max_{ \substack{\|\xi\|=1\\
 \|\xi\|_0\le Cs_l}} \frac{1}{NM}\sumi\sumj\xi' Z_{ijt}Z'_{ijt}(\widehat{\boldsymbol{\eta}} - \overline{\boldsymbol{\eta}})(\widehat{\boldsymbol{\eta}} - \overline{\boldsymbol{\eta}})'Z_{ijt}Z'_{ijt}\xi\\
 \lesssim&
  \max_{t\in [T]}\max_{ \substack{\|\xi\|=1\\
 \|\xi\|_0\le Cs_l}} \sqrt{\frac{1}{NM}\sumi\sumj\Big(\xi' Z_{ijt}Z'_{ijt}(\widehat{\boldsymbol{\eta}} - \overline{\boldsymbol{\eta}})\Big)^2} 
 \sqrt{\frac{1}{NM}\sumi\sumj
\Big( (\widehat{\boldsymbol{\eta}} - \overline{\boldsymbol{\eta}})'Z_{ijt}Z'_{ijt} \xi \Big)^2 }.
\end{align*}
Due to the sparsity of all the feasible $\xi$, we have $\|\xi\|_1\le \sqrt{s_l}\|\xi\|$. 
Thus, by Assumption \ref{a:moments}, Lemma \ref{lemma:rates_eta}, and Lemma \ref{lemma:regularized_events} with $\mu=C\sqrt{NM \log a}$ under Assumptions \ref{a:sparsity}, \ref{a:moments} (1), and \ref{a:sparse_eigenvalues}, we have
\begin{align}
 &\max_{t\in [T]}\max_{ \substack{\|\xi\|=1\\
 \|\xi\|_0\le Cs_l}} \frac{1}{NM}\sumi\sumj\Big(\xi' Z_{ijt}Z'_{ijt}(\widehat{\boldsymbol{\eta}} - \overline{\boldsymbol{\eta}})\Big)^2\nonumber\\
 \le& 
 \max_{t\in [T]}\max_{ \substack{\|\xi\|=1\\
 \|\xi\|_0\le Cs_l}} (\max_{i,j}|\xi' Z_{ijt}|^2) \frac{1}{NM}\sumi\sumj\Big(Z'_{ijt}(\widehat{\boldsymbol{\eta}} - \overline{\boldsymbol{\eta}})\Big)^2\nonumber\\
 \le&
  \max_{t\in [T]}\max_{ \substack{\|\xi\|=1\\
 \|\xi\|_0\le Cs_l}} (\max_{i,j}\|\xi'\|^2_1 \cdot( 1\vee\|X_{ijt}\|^2_\infty)) \frac{1}{NM} \|Z(\widehat{\boldsymbol{\eta}}-\overline{\boldsymbol{\eta}})\|^2\nonumber\\
 \lesssim&  
 s_l \cdot  O_p(1\vee E[\max_{i,j,t}\|X_{ijt} \|^2_\infty])O_p\Big( \frac{s\log a}{NM}\Big)
 = O_p\Big(\frac{s \cdot s_l B^2_{NM}\log a}{(NM)^{1-1/q}}
 \Big).\label{eq:variance_estimator_1} 
\end{align}
Therefore, $(8)= O_p\Big(\frac{s \cdot s_l B^2_{NM}\log a}{(NM)^{1-1/q}} \Big).$

Similarly, for $(1)$ and $(3)$, we have
\begin{align*}
& \max_{ \substack{\|\xi\|=1\\
 \|\xi\|_0\le Cs_l}}
 \xi'\sumi\sumj \frac{1}{NM}\Big\{ \Big(\sumt Z_{ijt}R_{ijt}\Big)\Big(\sumt  Z_{ijt}\varepsilon_{ijt} \Big)'\Big\}\xi\\
 \le 
 & T^2\max_{t\in [T]}\max_{ \substack{\|\xi\|=1\\
 \|\xi\|_0\le Cs_l}}\frac{1}{NM}
 \sumi\sumj \xi' Z_{ijt}R_{ijt}  \varepsilon_{ijt} Z'_{ijt}\xi\\
  \le 
 & T^2\max_{t\in [T]}\max_{ \substack{\|\xi\|=1\\
 \|\xi\|_0\le Cs_l}}
 \sqrt{
 \frac{1}{NM}
 \sumi\sumj (\xi' Z_{ijt}R_{ijt} )^2
 } 
 \sqrt{
 \frac{1}{NM}\sumi \sumj (\varepsilon_{ijt} Z'_{ijt}\xi)^2 
 }
\end{align*}
Thus, by Assumptions \ref{a:sparsity} and \ref{a:moments} (1),
\begin{align}
 \max_{ \substack{\|\xi\|=1\\
 \|\xi\|_0\le Cs_l}} \frac{1}{NM}
 \sumi\sumj (\xi' Z_{ijt}R_{ijt} )^2
 \le &
  \max_{ \substack{\|\xi\|=1\\
 \|\xi\|_0\le Cs_l}} \frac{1}{NM}\max_{i,j,t}|\xi' Z_{ijt}|\sumi\sumj R_{ijt}^2\nonumber\\
  \le&   
 \max_{ \substack{\|\xi\|=1\\
 \|\xi\|_0\le Cs_l}}  \frac{1}{NM}\|\xi\|_1^2\cdot (1 \vee \max_{i,j,t}\| Z_{ijt} \|_\infty^2) \cdot s\nonumber\\
  =&O_p\Big( \frac{s \cdot s_l B_{NM}^2}{(NM)^{1-1/q}}\Big)\label{eq:variance_estimator_2} 
\end{align}
for all feasible $\xi$.
Since $Z'Z/NM=Q \bar \Psi Q$ and $\|Q\xi\|\le\|\xi\|$,
\begin{align}
\max_{ \substack{\|\xi\|=1\\
 \|\xi\|_0\le Cs_l}} \frac{1}{NM}\sumi \sumj (\varepsilon_{ijt} Z'_{ijt}\xi)^2 
\le& 
\max_{i,j,t}|\varepsilon_{ijt}|^2 \max_{ \substack{\|\xi\|=1\\
 \|\xi\|_0\le Cs_l}} \frac{1}{NM}\sumi \sumj ( Z'_{ijt}\xi)^2\nonumber \\
\le&  
\max_{i,j,t}|\varepsilon_{ijt}|^2 \max_{ \substack{\|\xi\|=1\\
 \|\xi\|_0\le Cs_l}} \xi' Q\bar \Psi Q \xi\nonumber \\
 \le&  
\max_{i,j,t}|\varepsilon_{ijt}|^2 \max_{ \substack{\|\xi\|=1\\
 \|\xi\|_0\le Cs_l}} \xi' \bar \Psi  \xi\nonumber \\
\le&  
O_p\Big(E \max_{i,j,t}|\varepsilon_{ijt}|^2\Big) \varphi^2_{\max}(\bar\Psi,Cs_l) 
= 
O_p\Big( (NM)^{1/q}\Big),
\label{eq:variance_estimator_3} 
\end{align}
where the second inequality is due to Assumption \ref{a:sparse_eigenvalues} and 
the last uses Assumption \ref{a:moments} (3).

Since all the remaining terms consist of the products of the above three components, by using (\ref{eq:variance_estimator_1}), (\ref{eq:variance_estimator_2}) and (\ref{eq:variance_estimator_3}), we obtain
\begin{align*}
\max_{ \substack{\|\xi\|=1
 \|\xi\|_0\le Cs_l}}\xi'(\hat \Omega -\tilde \Omega )\xi\le&
 O_p\Big( \sqrt{\frac{s \cdot s_l B^2_{NM} }{(NM)^{1-2/q}} }\Big)
 +
 O_p\Big( \sqrt{\frac{s \cdot s_l B^2_{NM} \log a}{(NM)^{1-2/q}} }\Big)
 +
  O_p\Big( \sqrt{\frac{s \cdot s_l B^2_{NM} }{(NM)^{1-2/q}} }\Big)   \\
&+  
   O_p\Big( \frac{s \cdot s_l B^2_{NM} }{(NM)^{1-1/q}} \Big)
+
  O_p\Big( \frac{s \cdot s_l B^2_{NM}\sqrt{\log a} }{(NM)^{1-1/q}} \Big)
  +
   O_p\Big( \sqrt{\frac{s \cdot s_l B^2_{NM} \log a}{(NM)^{1-2/q}} }\Big)\\
&   +
     O_p\Big( \frac{s \cdot s_l B^2_{NM}\sqrt{\log a} }{(NM)^{1-1/q}} \Big)
     +   
     O_p\Big( \frac{s \cdot s_l B^2_{NM}\log a }{(NM)^{1-1/q}} \Big)\\
     = &O_p\Big( \sqrt{\frac{s \cdot s_l B^2_{NM} \log a}{(NM)^{1-2/q}} }\Big).
\end{align*}
Using the rate for $\max_{l\in [k_0]}\|\hat \Theta_{l}-\Theta_{l}\|$ from Lemma \ref{lemma:Theta_1_way} under Assumptions \ref{a:sparsity}, \ref{a:moments}, and \ref{a:variance}, we have
\begin{align*}
(\ref{eq:variance_est})=O_p\Big( \sqrt{\frac{s \cdot s_l B^2_{NM} \log a}{(NM)^{1-2/q}} }\Big) + \frac{s_l\log a}{NM} O(1) + O(1)O_p(1)O_p\Big(\sqrt{\frac{s_l\log a}{NM}}\Big) =o_P(1)
\end{align*}
as desired.
\end{proof}

\begin{remark}
As emphasized in the main text, recall that Assumption \ref{a:Theta} (4) requires $ss_l(\log (p\vee (NM)))^2/(N \wedge M)=o(1)$ instead of $ss^2_l(\log(p\vee (NM)))^2/(N \wedge M)=o(1)$. 
This is due to the fact that we made use of the bound 
$ |\hat\Theta_{l}' \hat\Omega \hat\Theta_{l}-\hat\Theta_{l}' \tilde\Omega \hat\Theta_{l}|\le \|\hat \Theta_l\|^2 \max_{ \substack{\|\xi\|=1\\
 \|\xi\|_0\le Cs_l}}\xi'(\hat\Omega - \tilde\Omega)\xi   $ with probability approaching unity following Lemma \ref{lemma:empirical_pre-sparsity}. On the other hand , in \cite{Kock2016} and \cite{KockTang2018}, the bound based on the dual norm inequality $|\hat\Theta_{l}' \hat\Omega \hat\Theta_{l}-\hat\Theta_{l}' \tilde\Omega \hat\Theta_{l}|\le \|\hat \Theta_l\|_1^2 \|\hat \Omega - \tilde \Omega\|_\infty$ is used in place.
\end{remark}

\section{Auxiliary Lemmas}
\label{sec:auxiliary_lemmas}

\subsection{Oracle Inequalities}
\label{sec:oracle_inequalities}

\begin{assumption}[Oracle Inequalities]\label{a:rates}
For each $(N,M)$ and for some choice of $\mu$ that depends on $(N,M)$,
we have $2\|\widehat \Upsilon_1^{-1}\varepsilon'X\|_\infty \le \mu/c$, $2\|\widehat \Upsilon_2^{-1}\varepsilon'D_1\|_\infty \le\mu/\sqrt{N}c$ and $2\|\widehat \Upsilon_3^{-1}\varepsilon'D_2\|_\infty \le \mu/\sqrt{M}c$ with probability $1-o(1)$ for some $c>1$.
\end{assumption}

\begin{assumption}[Weights for Penalty]\label{a:penalty_loading}
There exist the ideal penalty loading matrix $\widehat \Upsilon^0_l$ with all elements bounded and bounded away from zero uniformly over $(N,M)$, sequences $u$, $\ell$ with $0<\ell \le 1\le u$, $\ell \overset{p}{\to} 1$, and $u  \overset{p}{\to} u'>1$ for some constant $u'$ such that
\begin{align*}
\ell \widehat \Upsilon^0_l \le \widehat \Upsilon_l \le u\widehat\Upsilon^0_l
\end{align*}
with probability $1-o(1)$ for $l=1,2,3$.
\end{assumption}
\begin{remark}\label{rem:weights}
There are many possible situations where one may want to impose weights to penalize different parameters differently. These situations include (1) the case where one incorporates extra information from economic theory; (2) a penalty choice based on the theory of moderate deviation inequality for self-normalized sums as in \cite{BelloniChenChernozhukovHansen2012}; (3) the case where one conducts an iterating lasso algorithm such as the conservative lasso as in \cite{CanerKock2018}; and (4) the common practice of normalizing the standard errrs of all covariates to one.
\end{remark}

\begin{assumption}[Restricted Eigenvalues]\label{a:restricted_eigenvalue}
For any $C>0$, there exists $\underline\kappa_C>0$ depends only on $C$ such that $\kappa^2_{C}:=\kappa^2_{C}(\bar\Psi,s_1,s_2,s_3)\ge \underline\kappa_C$ for all $(N,M)$ with probability $1-o(1)$.
\end{assumption}
\begin{remark}
As highlighted in \cite{BelloniChenChernozhukovHansen2012}, Assumption \ref{a:sparse_eigenvalues} implies Assumption \ref{a:restricted_eigenvalue} by the argument in \cite{BickelRitovTsybakov2009}.
\end{remark}
The following lemma presents oracle inequalities for three-dimensional panel lasso. Its proof is closely related to Lemma 6 of \cite{BelloniChenChernozhukovHansen2012}. The main difference is that it accounts for the presence of fixed effects with different effective sample sizes.
\begin{lemma}[Oracle Inequalities]\label{lemma:rates_eta}
If Assumptions \ref{a:sparsity}, \ref{a:rates}, \ref{a:penalty_loading}, and \ref{a:restricted_eigenvalue} are satisfied, then
 \begin{align*}
 &\|Z(\widehat{\boldsymbol{\eta}}-\overline{\boldsymbol{\eta}})\|=(\widehat{\boldsymbol{\eta}}-\overline{\boldsymbol{\eta}})'Q\bar \Psi Q(\widehat{\boldsymbol{\eta}}-\overline{\boldsymbol{\eta}})
 \lesssim \frac{\mu\sqrt{s}}{\sqrt{NM}\kappa_{c_0}} + c_s,\\
&\sqrt{(\widehat{\boldsymbol{\eta}}-\overline{\boldsymbol{\eta}})'\frac{Z'Z}{NM}(\widehat{\boldsymbol{\eta}}-\overline{\boldsymbol{\eta}})}\lesssim \frac{\mu\sqrt{s}}{NM\kappa_{c_0}} + \frac{c_s}{\sqrt{NM}},\\
 &\|\widehat \Upsilon_1^0 (\widehat{\boldsymbol{\beta}} - {\boldsymbol{\beta}})\|_1\lesssim \frac{\mu s}{NM\kappa_{2c_0}\kappa_{c_0}}+ \frac{\sqrt{s}c_s}{\sqrt{NM}\kappa_{2c_0}} +  \frac{ c_s^2}{\mu},\\
  &\|\widehat \Upsilon_2^0 (\widehat{\boldsymbol{\alpha}} - \overline{\boldsymbol{\alpha}})\|_1\lesssim \frac{\mu s}{\sqrt{N}M\kappa_{2c_0}\kappa_{c_0}}+ \frac{\sqrt{ s}c_s}{\sqrt{M}\kappa_{2c_0}} +  \frac{N^{1/2} c_s^2}{\mu}, \qquad\text{and}\\
  &\|\widehat \Upsilon_3^0 (\widehat{\boldsymbol{\gamma}} - \overline{\boldsymbol{\gamma}})\|_1\lesssim \frac{\mu s}{N\sqrt{M}\kappa_{2c_0}\kappa_{c_0}}+ \frac{\sqrt{ s}c_s}{\sqrt{N}\kappa_{2c_0}} +  \frac{M^{1/2} c_s^2}{\mu}.
 \end{align*}
\end{lemma}

\begin{proof}
From the definition of $\widehat{\boldsymbol{\eta}}$, we have 
\begin{align*}
& \|y- Z \widehat{\boldsymbol{\eta}}\|^2 + \mu P(\widehat{\boldsymbol{\eta}})
\le
\|y- Z \overline{\boldsymbol{\eta}}\|^2 + \mu P(\overline{\boldsymbol{\eta}}).
\end{align*}
Rewrite this inequality and and get
\begin{align*}
& \| Z( \widehat{\boldsymbol{\eta}} - \overline{\boldsymbol{\eta}} ) + (R+\varepsilon)\|^2 + \mu P(\widehat{\boldsymbol{\eta}})
\le
\|(R+\varepsilon)\|^2 + \mu P(\overline{\boldsymbol{\eta}}).
\end{align*}
\begin{align*}
& \| Z( \widehat{\boldsymbol{\eta}} - \overline{\boldsymbol{\eta}} ) + (R+\varepsilon)\|^2 
\le
\|(R+\varepsilon)\|^2 + \mu (P(\overline{\boldsymbol{\eta}})- P(\widehat{\boldsymbol{\eta}}) ).
\end{align*}
Using reverse triangle inequality and the dual norm inequality,
\begin{align*}
 \| Z( \widehat{\boldsymbol{\eta}} - \overline{\boldsymbol{\eta}} )\|^2  
\le&
2 |\varepsilon'Z( \widehat{\boldsymbol{\eta}} - \overline{\boldsymbol{\eta}} )|+
2 |R' Z( \widehat{\boldsymbol{\eta}} - \overline{\boldsymbol{\eta}} )|
 + \mu \Big\{P( (\overline{\boldsymbol{\eta}} - \widehat{\boldsymbol{\eta}})_J )- P((\overline{\boldsymbol{\eta}} - \widehat{\boldsymbol{\eta}})_{J^c})  \Big\}\\
 \le&
2 \|(\widehat \Upsilon_1)^{-1}\varepsilon'X\|_\infty\|\widehat \Upsilon_1(\widehat{\boldsymbol{\beta}} - {\boldsymbol{\beta}}) \|_1+2 \|(\widehat \Upsilon_2)^{-1}\varepsilon'D_1\|_\infty\|\widehat \Upsilon_2( \widehat{\boldsymbol{\alpha}} - \overline{\boldsymbol{\alpha}} )\|_1 \\
&+ 2 \|(\widehat \Upsilon_3)^{-1}\varepsilon'D_2\|_\infty\|\widehat \Upsilon_3 \widehat{\boldsymbol{\gamma}} - \overline{\boldsymbol{\gamma}} \|_1\\
&+ 2 \|R\| \|Z( \widehat{\boldsymbol{\eta}} - \overline{\boldsymbol{\eta}} )\|
 + \mu \Big\{P( (\overline{\boldsymbol{\eta}} - \widehat{\boldsymbol{\eta}})_J )- P((\overline{\boldsymbol{\eta}} - \widehat{\boldsymbol{\eta}})_{J^c})  \Big\}\\
\le&
\frac{\mu}{c}\Big(\|\widehat \Upsilon_1(\widehat{\boldsymbol{\beta}} - {\boldsymbol{\beta}}) \|_1+\frac{1}{\sqrt{N}}\|\widehat \Upsilon_2( \widehat{\boldsymbol{\alpha}} - \overline{\boldsymbol{\alpha}}) \|_1 + \frac{1}{\sqrt{M}}\|\widehat \Upsilon_3( \widehat{\boldsymbol{\gamma}} - \overline{\boldsymbol{\gamma}}) \|_1\Big)\\
&+ 2 c_s \|Z( \widehat{\boldsymbol{\eta}} - \overline{\boldsymbol{\eta}} )\|
 + \mu \Big\{P( (\overline{\boldsymbol{\eta}} - \widehat{\boldsymbol{\eta}})_J )- P((\overline{\boldsymbol{\eta}} - \widehat{\boldsymbol{\eta}})_{J^c})  \Big\},
\end{align*}
where the third inequality follows from Assumptions \ref{a:sparsity} and \ref{a:rates}.
By the definition of $P$,  we have
\begin{align}
 \| Z( \widehat{\boldsymbol{\eta}} - \overline{\boldsymbol{\eta}} )\|^2  \le& 
\mu\Big(u+\frac{1}{c}\Big)\Big(\|\widehat \Upsilon^0_1(\widehat{\boldsymbol{\beta}} - {\boldsymbol{\beta}})_{J_1} \|_1+\frac{1}{\sqrt{N}}\|\widehat \Upsilon^0_2( \widehat{\boldsymbol{\alpha}} - \overline{\boldsymbol{\alpha}})_{J_2} \|_1 + \frac{1}{\sqrt{M}}\|\widehat \Upsilon^0_3( \widehat{\boldsymbol{\gamma}} - \overline{\boldsymbol{\gamma}})_{J_3} \|_1\Big)\nonumber\\
&
-\mu\Big(\ell-\frac{1}{c}\Big)\Big(\|\widehat \Upsilon^0_1(\widehat{\boldsymbol{\beta}} - {\boldsymbol{\beta}})_{J^c_1} \|_1+\frac{1}{\sqrt{N}}\|\widehat \Upsilon^0_2( \widehat{\boldsymbol{\alpha}} - \overline{\boldsymbol{\alpha}})_{J^c_2} \|_1 + \frac{1}{\sqrt{M}}\|\widehat \Upsilon^0_3( \widehat{\boldsymbol{\gamma}} - \overline{\boldsymbol{\gamma}})_{J^c_3} \|_1\Big)\nonumber\\
&+ 2c_s \|Z( \widehat{\boldsymbol{\eta}} - \overline{\boldsymbol{\eta}} )\|\label{eq:basic_inequality}
\end{align}
under Assumption \ref{a:penalty_loading}.

We now branch into two cases.
First, suppose that $\|Z(\widehat{\boldsymbol{\eta}} - \overline{\boldsymbol{\eta}})\| < 2 c_s$.
In this case, the first equation in the statement of the lemma is trivially true since all the terms on right-hand side of the first equation in the statement of the lemma are non-negative. 
Second, suppose that $\|Z(\widehat{\boldsymbol{\eta}} - \overline{\boldsymbol{\eta}})\|\ge 2c_s$.
In this case,
\begin{align*}
 \| Z( \widehat{\boldsymbol{\eta}} - \overline{\boldsymbol{\eta}} )\|^2  \le& 
\mu\Big(u+\frac{1}{c}\Big)\Big(\|\widehat \Upsilon^0_1(\widehat{\boldsymbol{\beta}} - {\boldsymbol{\beta}})_{J_1} \|_1+\frac{1}{\sqrt{N}}\|\widehat \Upsilon^0_2( \widehat{\boldsymbol{\alpha}} - \overline{\boldsymbol{\alpha}})_{J_2} \|_1 + \frac{1}{\sqrt{M}}\|\widehat \Upsilon^0_3( \widehat{\boldsymbol{\gamma}} - \overline{\boldsymbol{\gamma}})_{J_3} \|_1\Big)\\
&
-\mu\Big(\ell-\frac{1}{c}\Big)\Big(\|\widehat \Upsilon^0_1(\widehat{\boldsymbol{\beta}} - {\boldsymbol{\beta}})_{J^c_1} \|_1+\frac{1}{\sqrt{N}}\|\widehat \Upsilon^0_2( \widehat{\boldsymbol{\alpha}} - \overline{\boldsymbol{\alpha}})_{J^c_2} \|_1 + \frac{1}{\sqrt{M}}\|\widehat \Upsilon^0_3( \widehat{\boldsymbol{\gamma}} - \overline{\boldsymbol{\gamma}})_{J^c_3} \|_1\Big)\\
&+ \|Z( \widehat{\boldsymbol{\eta}} - \overline{\boldsymbol{\eta}} )\|^2,
\end{align*}
and thus 
\begin{align}
&\Big(\|\widehat \Upsilon^0_1(\widehat{\boldsymbol{\beta}} - {\boldsymbol{\beta}})_{J^c_1} \|_1+\frac{1}{\sqrt{N}}\|\widehat \Upsilon^0_2( \widehat{\boldsymbol{\alpha}} - \overline{\boldsymbol{\alpha}})_{J^c_2} \|_1 + \frac{1}{\sqrt{M}}\|\widehat \Upsilon^0_3( \widehat{\boldsymbol{\gamma}} - \overline{\boldsymbol{\gamma}})_{J^c_3} \|_1\Big)\nonumber \\ 
\le& 
c_0\Big(\|\widehat \Upsilon^0_1(\widehat{\boldsymbol{\beta}} - {\boldsymbol{\beta}})_{J_1} \|_1+\frac{1}{\sqrt{N}}\|\widehat \Upsilon^0_2( \widehat{\boldsymbol{\alpha}} - \overline{\boldsymbol{\alpha}})_{J_2} \|_1 + \frac{1}{\sqrt{M}}\|\widehat \Upsilon^0_3( \widehat{\boldsymbol{\gamma}} - \overline{\boldsymbol{\gamma}})_{J_3} \|_1\Big)\label{eq:restricted_event},
\end{align}
where $c_0=(uc+1)/(\ell c-1)$. 
Assumption \ref{a:restricted_eigenvalue} implies that, for any $\delta$ which is in the choice set of the minimum of restricted eigenvalue definition, we have
\begin{align*}
\kappa^2_{c_0}=\underset{ \substack{
R_1 \subset [k],\, |R_1|\le s_1\\
R_2 \subset [N_0],\, |R_2|\le s_2\\
R_3 \subset [M_0],\, |R_3|\le s_3\\
R= R_1 \cup R_2 \cup R_3 
} }{\min} 
\underset{\substack{\delta\in \mathbb{R}^p\setminus \{0\} \\
\|\delta_{J}^c\|_1 \le C\|\delta_{J}\|_1
}
}{\min} (s_1+s_2+s_3)\frac{\delta'\bar\Psi\delta}{\|\delta\|^2_1}.
\end{align*}
Since $\delta'\Psi\delta=\delta'S^{-1} Z'Z S^{-1}\delta=b'Z'Zb$ for $b=S^{-1}\delta$, we can rewrite the condition in terms of $b$ and obtain
\begin{align*}
\kappa^2_{c_0}=\underset{ \substack{
R_1 \subset [k],\, |R_1|\le s_1\\
R_2 \subset [N_0],\, |R_2|\le s_2\\
R_3 \subset [M_0],\, |R_3|\le s_3\\
R= R_1 \cup R_2 \cup R_3 
} }{\min} 
\underset{\substack{b\in \mathbb{R}^p\setminus \{0\} \\
\|b^1_{R^c_1}\|_1 +\frac{1}{\sqrt{N}}\|b^2_{R^c_2}\|_1
+\frac{1}{\sqrt{M}}\|b^3_{R^c_3}\|_1
\\
 \le c_0 \|b_{R_2}\|_1
 +\frac{1}{\sqrt{N}}\|b^2_{R_2}\|_1
 +\frac{1}{\sqrt{M}}\|b^3_{R_3}\|_1
}
}{\min} (s_1+s_2+s_3)\frac{\|Zb\|^2}{NM\|(NM)^{-1} S b \|^2_1}
\end{align*}
Note that (\ref{eq:restricted_event}) implies that we can let $b=\widehat{\boldsymbol{\eta}}-\overline{\boldsymbol{\eta}}$. 
Thus,
\begin{align*}
\Bigg\|\substack{ (\widehat{\boldsymbol{\beta}} - {\boldsymbol{\beta}})_{J_1}\\
\frac{1}{\sqrt{N}}(\widehat{\boldsymbol{\alpha}} -\overline{\boldsymbol{\alpha}})_{J_2}\\
\frac{1}{\sqrt{M}}(\widehat{\boldsymbol{\gamma}} -\overline{\boldsymbol{\gamma}})_{J_3}
 }\Bigg\|_1^2\le 
\frac{(s_1+ s_2 + s_3)} {\kappa^2_{c_0} NM}
  \|Z(\widehat{\boldsymbol{\eta}} -\overline{\boldsymbol{\eta}})\|^2.
\end{align*}
Taking the square root on both sides yields
\begin{align}
\|\widehat \Upsilon^0_1(\widehat{\boldsymbol{\beta}} - {\boldsymbol{\beta}})_{J_1} \|_1+\frac{1}{\sqrt{N}}\|\widehat \Upsilon^0_2( \widehat{\boldsymbol{\alpha}} - \overline{\boldsymbol{\alpha}})_{J_2} \|_1 + \frac{1}{\sqrt{M}}\|\widehat \Upsilon^0_3( \widehat{\boldsymbol{\gamma}} - \overline{\boldsymbol{\gamma}})_{J_3} \|_1
\le
 \frac{\sqrt{s_1+ s_2 + s_3}} {\kappa_{c_0} \sqrt{NM}}
  \|Z(\widehat{\boldsymbol{\eta}} -\overline{\boldsymbol{\eta}})\|. \label{eq:RE_inequalty}
\end{align}
Finally, substitute this equation into (\ref{eq:basic_inequality}) and drop the negative terms on the right-hand side yield
\begin{align*}
 \| Z( \widehat{\boldsymbol{\eta}} - \overline{\boldsymbol{\eta}} )\|  \le& 
\mu\Big(u+\frac{1}{c}\Big) \frac{\sqrt{s_1+ s_2 + s_3}} {\kappa_{c_0} \sqrt{NM}}
 + 2c_s.
\end{align*}
This shows the first equation in the statement of the lemma.

We next obtain the $L^1$-norm bounds.
We branch into two cases.
First, suppose that
\begin{align*}
&\Big(\|\widehat \Upsilon^0_1(\widehat{\boldsymbol{\beta}} - {\boldsymbol{\beta}})_{J^c_1} \|_1+\frac{1}{\sqrt{N}}\|\widehat \Upsilon^0_2( \widehat{\boldsymbol{\alpha}} - \overline{\boldsymbol{\alpha}})_{J^c_2} \|_1 + \frac{1}{\sqrt{M}}\|\widehat \Upsilon^0_3( \widehat{\boldsymbol{\gamma}} - \overline{\boldsymbol{\gamma}})_{J^c_3} \|_1\Big) \\ \le& 
2c_0\Big(\|\widehat \Upsilon^0_1(\widehat{\boldsymbol{\beta}} - {\boldsymbol{\beta}})_{J_1} \|_1+\frac{1}{\sqrt{N}}\|\widehat \Upsilon^0_2( \widehat{\boldsymbol{\alpha}} - \overline{\boldsymbol{\alpha}})_{J_2} \|_1 + \frac{1}{\sqrt{M}}\|\widehat \Upsilon^0_3( \widehat{\boldsymbol{\gamma}} - \overline{\boldsymbol{\gamma}})_{J_3} \|_1\Big).
\end{align*}
By definition of $\kappa_{2 c_0}$, we have
\begin{align*}
&\Big(\|\widehat \Upsilon^0_1(\widehat{\boldsymbol{\beta}} - {\boldsymbol{\beta}}) \|_1+\frac{1}{\sqrt{N}}\|\widehat \Upsilon^0_2( \widehat{\boldsymbol{\alpha}} - \overline{\boldsymbol{\alpha}}) \|_1 + \frac{1}{\sqrt{M}}\|\widehat \Upsilon^0_3( \widehat{\boldsymbol{\gamma}} - \overline{\boldsymbol{\gamma}}) \|_1\Big) \\ \le& 
(1+2c_0)\Big(\|\widehat \Upsilon^0_1(\widehat{\boldsymbol{\beta}} - {\boldsymbol{\beta}})_{J_1} \|_1+\frac{1}{\sqrt{N}}\|\widehat \Upsilon^0_2( \widehat{\boldsymbol{\alpha}} - \overline{\boldsymbol{\alpha}})_{J_2} \|_1 + \frac{1}{\sqrt{M}}\|\widehat \Upsilon^0_3( \widehat{\boldsymbol{\gamma}} - \overline{\boldsymbol{\gamma}})_{J_3} \|_1\Big)\\
\le& (1+2c_0)\frac{\sqrt{s_1 + s_2 + s_3}}{\kappa_{2c_0}\sqrt{NM}}\|Z(\widehat{\boldsymbol{\eta}} - \overline{\boldsymbol{\eta}})\|
\end{align*}
by applying similar lines of arguments to those of the first part of the proof using $2c_0$ in place of $c_0$.
Second, suppose that 
\begin{align}
&\Big(\|\widehat \Upsilon^0_1(\widehat{\boldsymbol{\beta}} - {\boldsymbol{\beta}})_{J^c_1} \|_1+\frac{1}{\sqrt{N}}\|\widehat \Upsilon^0_2( \widehat{\boldsymbol{\alpha}} - \overline{\boldsymbol{\alpha}})_{J^c_2} \|_1 + \frac{1}{\sqrt{M}}\|\widehat \Upsilon^0_3( \widehat{\boldsymbol{\gamma}} - \overline{\boldsymbol{\gamma}})_{J^c_3} \|_1\Big) \nonumber\\ >& 
2c_0\Big(\|\widehat \Upsilon^0_1(\widehat{\boldsymbol{\beta}} - {\boldsymbol{\beta}})_{J_1} \|_1+\frac{1}{\sqrt{N}}\|\widehat \Upsilon^0_2( \widehat{\boldsymbol{\alpha}} - \overline{\boldsymbol{\alpha}})_{J_2} \|_1 + \frac{1}{\sqrt{M}}\|\widehat \Upsilon^0_3( \widehat{\boldsymbol{\gamma}} - \overline{\boldsymbol{\gamma}})_{J_3} \|_1\Big).\label{eq:ineq_for_l_1}.
\end{align}
In this case, equation (\ref{eq:basic_inequality}) implies 
\begin{align*}
 \| Z( \widehat{\boldsymbol{\eta}} - \overline{\boldsymbol{\eta}} )\|^2  \le& 
\mu\Big(u+\frac{1}{c}\Big)\Big(\|\widehat \Upsilon^0_1(\widehat{\boldsymbol{\beta}} - {\boldsymbol{\beta}})_{J_1} \|_1+\frac{1}{\sqrt{N}}\|\widehat \Upsilon^0_2( \widehat{\boldsymbol{\alpha}} - \overline{\boldsymbol{\alpha}})_{J_2} \|_1 + \frac{1}{\sqrt{M}}\|\widehat \Upsilon^0_3( \widehat{\boldsymbol{\gamma}} - \overline{\boldsymbol{\gamma}})_{J_3} \|_1\Big)\nonumber\\
&
-\mu\Big(\ell-\frac{1}{c}\Big)\Big(\|\widehat \Upsilon^0_1(\widehat{\boldsymbol{\beta}} - {\boldsymbol{\beta}})_{J^c_1} \|_1+\frac{1}{\sqrt{N}}\|\widehat \Upsilon^0_2( \widehat{\boldsymbol{\alpha}} - \overline{\boldsymbol{\alpha}})_{J^c_2} \|_1 + \frac{1}{\sqrt{M}}\|\widehat \Upsilon^0_3( \widehat{\boldsymbol{\gamma}} - \overline{\boldsymbol{\gamma}})_{J^c_3} \|_1\Big)\nonumber\\
&+ 2c_s \|Z( \widehat{\boldsymbol{\eta}} - \overline{\boldsymbol{\eta}} )\|
\le 2c_s \|Z( \widehat{\boldsymbol{\eta}} - \overline{\boldsymbol{\eta}} )\|,
\end{align*}
where the last inequality is due to the definition of $c_0 = (uc + 1)/(\ell c - 1)$.
Equation (\ref{eq:basic_inequality}) further implies that
\begin{align*}
&\Big(\|\widehat \Upsilon^0_1(\widehat{\boldsymbol{\beta}} - {\boldsymbol{\beta}})_{J^c_1} \|_1+\frac{1}{\sqrt{N}}\|\widehat \Upsilon^0_2( \widehat{\boldsymbol{\alpha}} - \overline{\boldsymbol{\alpha}})_{J^c_1} \|_1 + \frac{1}{\sqrt{M}}\|\widehat \Upsilon^0_3( \widehat{\boldsymbol{\gamma}} - \overline{\boldsymbol{\gamma}})_{J^c_3} \|_1\Big)\\ \le& 
c_0\Big(\|\widehat \Upsilon^0_1(\widehat{\boldsymbol{\beta}} - {\boldsymbol{\beta}})_{J_1} \|_1+\frac{1}{\sqrt{N}}\|\widehat \Upsilon^0_2( \widehat{\boldsymbol{\alpha}} - \overline{\boldsymbol{\alpha}})_{J_2} \|_1 + \frac{1}{\sqrt{M}}\|\widehat \Upsilon^0_3( \widehat{\boldsymbol{\gamma}} - \overline{\boldsymbol{\gamma}})_{J_3} \|_1\Big)\nonumber\\
&+ \frac{c}{\ell c -1} \frac{1}{\mu}\|Z( \widehat{\boldsymbol{\eta}} - \overline{\boldsymbol{\eta}} )\|(2c_s-\|Z( \widehat{\boldsymbol{\eta}} - \overline{\boldsymbol{\eta}} )\| )\\
\le& 
c_0\Big(\|\widehat \Upsilon^0_1(\widehat{\boldsymbol{\beta}} - {\boldsymbol{\beta}})_{J_1} \|_1+\frac{1}{\sqrt{N}}\|\widehat \Upsilon^0_2( \widehat{\boldsymbol{\alpha}} - \overline{\boldsymbol{\alpha}})_{J_2} \|_1 + \frac{1}{\sqrt{M}}\|\widehat \Upsilon^0_3( \widehat{\boldsymbol{\gamma}} - \overline{\boldsymbol{\gamma}})_{J_3} \|_1\Big)\nonumber\\
&+ \frac{c}{\ell c -1} \frac{1}{\mu}c_s^2\\
\le&
\frac{c_0}{2}\Big(\|\widehat \Upsilon^0_1(\widehat{\boldsymbol{\beta}} - {\boldsymbol{\beta}})_{J^c_1} \|_1+\frac{1}{\sqrt{N}}\|\widehat \Upsilon^0_2( \widehat{\boldsymbol{\alpha}} - \overline{\boldsymbol{\alpha}})_{J^c_2} \|_1 + \frac{1}{\sqrt{M}}\|\widehat \Upsilon^0_3( \widehat{\boldsymbol{\gamma}} - \overline{\boldsymbol{\gamma}})_{J^c_3} \|_1\Big)\nonumber\\
&+ \frac{c}{\ell c -1} \frac{1}{\mu}c_s^2,
\end{align*}
where 
the first inequality follows from (\ref{eq:basic_inequality}),
the second inequality follows from $\|Z( \widehat{\boldsymbol{\eta}} - \overline{\boldsymbol{\eta}} )\|(2c_s-\|Z( \widehat{\boldsymbol{\eta}} - \overline{\boldsymbol{\eta}} )\| )\le \max_{x\ge 0}x(2c_s-x)\le c_s^2$, and 
the third inequality follows from (\ref{eq:ineq_for_l_1}). 
Therefore,
\begin{align*}
&\Big(\|\widehat \Upsilon^0_1(\widehat{\boldsymbol{\beta}} - {\boldsymbol{\beta}}) \|_1+\frac{1}{\sqrt{N}}\|\widehat \Upsilon^0_2( \widehat{\boldsymbol{\alpha}} - \overline{\boldsymbol{\alpha}}) \|_1 + \frac{1}{\sqrt{M}}\|\widehat \Upsilon^0_3( \widehat{\boldsymbol{\gamma}} - \overline{\boldsymbol{\gamma}}) \|_1\Big)\\
\le &
\Big(1+\frac{1}{2c_0}\Big)\Big(\|\widehat \Upsilon^0_1(\widehat{\boldsymbol{\beta}} - {\boldsymbol{\beta}})_{J^c_1} \|_1+\frac{1}{\sqrt{N}}\|\widehat \Upsilon^0_2( \widehat{\boldsymbol{\alpha}} - \overline{\boldsymbol{\alpha}})_{J^c_2} \|_1 + \frac{1}{\sqrt{M}}\|\widehat \Upsilon^0_3( \widehat{\boldsymbol{\gamma}} - \overline{\boldsymbol{\gamma}})_{J^c_3} \|_1\Big)\\
\le& \Big(1+\frac{1}{2c_0}\Big)\frac{2c}{\ell c-1}\frac{1}{\mu}c^2_s,
\end{align*}
where
the first inequality is due to (\ref{eq:ineq_for_l_1})
and
the second inequality is due to the previous equation.
Combining the two cases together, we obtain
\begin{align*}
&\Big(\|\widehat \Upsilon^0_1(\widehat{\boldsymbol{\beta}} - {\boldsymbol{\beta}}) \|_1+\frac{1}{\sqrt{N}}\|\widehat \Upsilon^0_2( \widehat{\boldsymbol{\alpha}} - \overline{\boldsymbol{\alpha}}) \|_1 + \frac{1}{\sqrt{M}}\|\widehat \Upsilon^0_3( \widehat{\boldsymbol{\gamma}} - \overline{\boldsymbol{\gamma}}) \|_1\Big)\\
\le& 
(1+2c_0)\frac{\sqrt{s_1 + s_2 + s_3}}{\kappa_{2c_0}\sqrt{NM}}\|Z(\widehat{\boldsymbol{\eta}} - \overline{\boldsymbol{\eta}})\|
+
\Big(1+\frac{1}{2c_0}\Big)\frac{2c}{\ell c-1}\frac{1}{\mu}c^2_s,
\end{align*}
and the remaining three equations in the statement of the lemma follow.
\end{proof}

\subsection{Concentration Inequality}
\label{sec:concentration_inequality}

The following lemma follows from \citet{CCK14} and \citet{CCK15}.

\begin{lemma}[A Concentration Inequality]\label{lemma:concentration_inequality_CCK}
Let $(X_i)_{i\in [n]}$ be $p$-dimensional independent random vectors,
$B=\sqrt{E[\max_{i\in [n]}\|X_i\|^2_\infty] }$, and 
$\sigma^2=\max_{j\in[p]}\frac{1}{n}\sum_{i=1}^n E|X_{ij}|^2$. 
With probability at least $1-C(\log n)^{-1}$,
\begin{align*}
 \max_{j\in [p]}\Big| \frac{1}{n}\sum_{i=1}^n (|X_{ij}| - E|X_{ij}|)\Big| \lesssim \sqrt{\frac{\sigma^2\log (p\vee n)}{n} }
+
\frac{B \log( p\vee n)}{n}.
\end{align*}
\end{lemma}
\begin{proof}
The claim follows from applying Theorem 5.1 of \citet{CCK14} to Lemma 8 of \citet{CCK15} with $t=\log n$, $\alpha= 1$, and $q=2$.
\end{proof}

\subsection{Regularized Events}
\label{sec:regularized_events}

\begin{lemma}[Regularized Events]\label{lemma:regularized_events}
Fix constants $c>1$ and $C > 0$, and let $\widehat\Upsilon=I$. 
If Assumption \ref{a:moments} is satisfied, then we have $2\|\varepsilon'X\|_\infty \le \mu/c$, $2\|\varepsilon'D_1\|_\infty \le\mu/c\sqrt{N}$ and $2\|\varepsilon'D_2\|_\infty \le \mu/c\sqrt{M}$ with probability at least $1-C(\log (N\wedge M))^{-1}$ where $\mu = C\sqrt{NM\log a}$. 
Similarly, if Assumption \ref{a:moments} is satisfied, then we have $\| X_{-l}\zeta_l\|_\infty \le \mu_{\text{node,l}}/2c$,  $\|D_l\zeta_l\|_\infty \le \mu_{\text{node,l}}/2c\sqrt{N}$,  and $\|D_2\zeta_l\|_\infty \le \mu_{\text{node,l}}/2c\sqrt{M}$ uniformly over $l\in [p]$ with probability at least $1-C(\log (N\wedge M))^{-1}$ where $\mu_{\text{node},l}=C\sqrt{NM\log a}$.
\end{lemma}

\begin{proof}
Applying Lemma \ref{lemma:concentration_inequality_CCK}, we have 
\begin{align*}
\frac{\|X'\varepsilon\|_\infty}{NM}=\max_{l\in[k_0]}\Big|\frac{1}{NM}\sumi\sumj\sumt (X_{ijt,l}\varepsilon_{ijt}-E[X_{ijt,l}\varepsilon_{ijt}])\Big|
\\
\lesssim \sqrt{\frac{\sigma^2 \log (p\vee(NM))}{NM}} + \frac{B\log (p\vee(NM))}{NM}
\end{align*}
with probability $1-C(\log NM)^{-1}$, where $\sigma^2= \max_{l\in[k_0]}\max_{t\in[T]}\frac{1}{NM}E[X^2_{ijt,l}\varepsilon_{ijt}^2]\le O(K^4)$.
Note that we have
\begin{align*}
B^2=&E[\max_{i\in [N],j\in[ M],t\in[T] }\|X_{ijt}\varepsilon_{ijt}\|^2_\infty]\\
\le&(E[\max_{i\in [N],j\in[ M],t\in[T] }\|X_{ijt}\|^q_\infty |\varepsilon_{ijt}|^q])^{2/q} \\
\lesssim& (NM)^{2/q}(E[\frac{1}{NM}\sumi\sumj\sumt\|X_{ijt}\|^q_\infty |\varepsilon_{ijt}|^q])^{2/q}\\
\lesssim& (NM)^{2/q}\Big[\Big(E[\frac{1}{NM}\sumi\sumj\sumt\|X_{ijt}\|^{2q}_\infty] \Big)^{1/2}\Big(E[\frac{1}{NM}\sumi\sumj\sumt |\varepsilon_{ijt}|^{2q}]\Big)^{1/2}\Big]^{2/q}\\
=&O((NM)^{2/q}B^2_{NM})
\end{align*}
where 
the first inequality is due to Jensen's inequality, 
the third inequality is due to H\"older's inequality, and
the last equality is due to Assumption \ref{a:moments} (1) and (3).
Thus $\frac{B_{NM}\log (p\vee (NM))}{(NM)^{1-1/q}}=O(\sqrt{\frac{\log a}{NM}})$,
and this implies
\begin{align*}
2c\|X'\varepsilon\|_\infty=\max_{l\in[k_0]}\Big|\sumi\sumj\sumt (X_{ijt,l}\varepsilon_{ijt}-E[X_{ijt,l}\varepsilon_{ijt}])\Big|\lesssim \sqrt{NM\log a}=C^{-1}\mu
\end{align*}
with probability at least $1-C(\log (N M))^{-1}$ for $K>0$ large enough. 

Since $\|(D_1,D_2)\|_\infty=1$ under Assumption \ref{a:moments} (2), an application of Lemma \ref{lemma:concentration_inequality_CCK} gives
\begin{align*}
2c\|D_1'\varepsilon\|_\infty=\max_{l\in \{k_0+1,...,k_0+N_0\}}\Big|
\sumj\sumt (D_{1,ijt,l}-ED_{1,ijt,l})
\Big|\lesssim \sqrt{ M\log a}=C^{-1}\mu/\sqrt{N}
\end{align*}
with probability at least $1-C(\log (N\wedge M))^{-1}$,
where $i$ depends on the choice of $l$.
Note that there are at most $MT=O(M)$ nonzero terms in the summand for each $l$. 
Analogous arguments hold for $\|D_2'\varepsilon\|_\infty$ with the number of nonzero terms being at most $NT=O(N)$ in place of $MT$. 

Under Assumption \ref{a:moments} and the choice $\mu_{\text{node},l}=C\sqrt{NM\log a}$, similar lines of argument to those above show that the regularized events $\| X_{-l}\zeta_l\|_\infty \le \mu_{\text{node}}/2c$,
 $\|D_l\zeta_l\|_\infty \le \mu_{\text{node}}/2c\sqrt{N}$, and
 $\|D_2\zeta_l\|_\infty \le \mu_{\text{node}}/2c\sqrt{M}$
occur with probability approaching one. 
Applying Lemma \ref{lemma:concentration_inequality_CCK}, we have 
\begin{align*}
&\frac{\|\zeta^{l \prime}Z_{-l}'\|_\infty}{NM}=\max_{k\in[k_0], l\in[p]}\Big|\frac{1}{NM}\sumi\sumj\sumt (Z^l_{ijt,k}\zeta^l_{ijt}-E[Z^l_{ijt,k}\zeta^l_{ijt}])\Big|\\
\lesssim& \sqrt{\frac{\sigma^2 \log (p^2\vee(NM))}{NM}} + \frac{B\log (p^2\vee(NM))}{NM}\lesssim \frac{\log a}{NM}
\end{align*}
with probability $1-C(\log N\wedge  M)^{-1}$.
\end{proof}

\subsection{Rates of Nuisance Parameters}
\label{sec:rates_of_nuisance_parameters}

Throughout this section, we use the following notations.
For any diagonal matrix $A$, $A_{l}$ denotes the $l$-th diagonal entry and $A_{-l}$ denotes $A$ with the $l$-th column and row removed. 

The following lemma establishes behaviors of the nuisance parameters based on the nodewise regressions under three-dimensional panel setting.
It is closely related to Lemma C.9 of \cite{KockTang2018}. 
The main difference is that, in \cite{KockTang2018}, their one-way fixed effect modeling assumption implies their $D_2=\emptyset$ and $D'_1D_1=I$, which in turn implies the diagonal structure of
$$\Theta=
\begin{bmatrix}
\Theta_X & 0\\
0& I
\end{bmatrix},
$$
and greatly simplifies their estimation procedure. 
In our case, however, due to the potential presence of multi-way fixed effects, such decomposition is not available. Therefore, the theory of our nodewise regression needs to account for these fixed effects with different convergence rates simultaneously.
\begin{lemma}[Nodewise Lasso for Nuisance Parameters]\label{lemma:Theta_1_way}
Suppose Assumptions \ref{a:moments}, \ref{a:sparse_eigenvalues}, and \ref{a:Theta} are satisfied and $\widehat \Theta$ is calculated following (\ref{eq:theta_hat}) with $\mu_{\text{node},l}=C\sqrt{NM \log a}$ for a $C>0$.
It holds uniformly over $l\in[k_0]$ that
\begin{align*}
\|\widehat \phi^l -\phi^l\|_1
 =&O_p\Big(\sqrt{\frac{s^2_l\log a}{N\wedge M}}
 \Big),\\
\|\widehat \phi^l -\phi^l\|
 =&O_p\Big(\sqrt{\frac{s_l\log a}{N\wedge M}}
 \Big),
\end{align*}
\begin{align*}
|\widehat \tau_l^2 - \tau_l^2| =&O_p\Big(\sqrt{\frac{s_l\log a}{NM}}\Big),\\
\Big|\frac{1}{\widehat \tau_l^2} - \frac{1}{\tau_l^2}\Big| =&O_p\Big(\sqrt{\frac{s_l\log a}{NM}}\Big),
\end{align*}
\begin{align*}
 \|\widehat\Theta_l'-\Theta_l'\|_1
 =&O_p\Big(\sqrt{\frac{s^2_l\log a}{N\wedge M}}
 \Big),\\
  \|\widehat\Theta_l'-\Theta_l'\|
 =&O_p\Big(\sqrt{\frac{s_l\log a}{N\wedge M}}
 \Big),\\
  \|\widehat \Theta_{l}\|_1=&O_p(s^{1/2}_l),\qquad\text{and}\\
\max_{l\in[k_0]}\Big|\frac{1}{\widehat \tau^2_l}\Big|=&O_p(1).
\end{align*}
\end{lemma}
\begin{proof}
The proof is consists of three steps.\\
\noindent 
\textbf{Step 1}
First, under Assumption \ref{a:moments} and by the choice $\mu_{\text{node},l}=C\sqrt{NM\log a}$, Lemma \ref{lemma:regularized_events} gives that the regularized events $\| X_{-l}\zeta_l\|_\infty \le \mu_{\text{node}}/2c$,
 $\|D_l\zeta_l\|_\infty \le \mu_{\text{node}}/2c\sqrt{N}$,
 $\|D_2\zeta_l\|_\infty \le \mu_{\text{node}}/2c\sqrt{M}$
occur with probability approaching one uniformly over $[k_0]$. 
Using the arguments similar to those of Lemma \ref{lemma:rates_eta}, under Assumptions \ref{a:sparse_eigenvalues} and \ref{a:Theta} (1) and (2), we have
\begin{align}
\frac{1}{NM}\|Z_{-l}(\widehat\phi^l - \phi^l)\|^2&= \,(\widehat\phi^l - \phi^l)'Q_{-l}\bar \Psi_{-l,-l}Q_{-l}(\widehat \phi^l -\phi^l) =
O_p\Big( \frac{s_l\log a}{NM}\Big),\label{eq:phi_prediction_norm}\\
\|Q_{-l}(\widehat \phi^l - \widehat \phi^l)\|_1
 =&
 O_p\Big( s_l\sqrt{\frac{\log a}{N M} }\Big),\label{eq:phi_l1_norm_weighted}\\
\|\widehat \phi^l - \widehat \phi^l\|_1
 =&O_p\Big( s_l\sqrt{\frac{\log a}{N\wedge M}}\Big).\label{eq:phi_l1_norm}
\end{align}
uniformly for $l \in [k_0]$.

To find a bound for $ \|\widehat\phi^l - \phi^l\|$ that holds uniformly over $[k_0]$, note that
\begin{align}
(\widehat \phi^l -\phi^l )'Q_{-l}\Psi_{-l,-l} Q_{-l}(\widehat \phi^l -\phi^l )\le&  
  (\widehat \phi^l -\phi^l )'Q_{-l}\bar \Psi_{-l,-l} Q_{-l}(\widehat \phi^l -\phi^l ) 
  +
 \|\bar \Psi - \Psi\|_\infty\|Q_{-l}(\widehat \phi^l -\phi^l) \|_1^2.\nonumber\\
  \le&  O_p\Big( \frac{s\log a}{NM}\Big)+  \|\bar \Psi - \Psi\|_\infty\|Q_{-l}(\widehat \phi^l -\phi^l )\|_1^2,\label{eq:psi_l2}
\end{align}
by (\ref{eq:phi_prediction_norm}), where $\|A\|_\infty$ denotes the maximal element of a matrix $A$.
We now bound the second term on the right-hand side.
Note that
\begin{align*}
P\Big( \|\bar \Psi - \Psi\|_\infty\ge r  \Big)
\le& 
P\Big( \max_{t\in [T]}\max_{l\in[k_0]}\Big\|\frac{1}{NM}\sum_{i,j}( X^2_{ijt,l}  - EX^2_{ijt,l}) \Big\|_\infty\ge r/T  \Big)\\
+& 
P\Big( \max_{t\in [T]}\max_{l\in\{k_0+1,...,k_0+N_0\}}\Big\|\frac{1}{M}\sum_{j}( D^2_{1,ijt,l}  - ED^2_{1,ijt,l}) \Big\|_\infty\ge r/T  \Big)\\
+&
 P\Big( \max_{t\in [T]}\max_{l\in\{k_0+N_0+1,...,k_0+N_0+M_0\}}\Big\|\frac{1}{N}\sum_{i}( D^2_{2,ijt,l}  - ED^2_{2,ijt,l}) \Big\|_\infty\ge r/T  \Big)
\end{align*}
We want to show all three terms go to zero with $r=C\sqrt{\frac{\log a}{NM}}$.
Assumption \ref{a:moments} (1) and (3) imply
\begin{align*}
B^2= E[\max_{i\le N,\, j\le M,\, t\le T}\|X_{ijt}\|^4_\infty]
 \le&
  ( E[\max_{i\le N,\, j\le M,\, t\le T}\|X_{ijt}\|^{2q}_\infty])^{2/q}\\
  \le&  (NM)^{2/q}\Big( E\Big[\frac{1}{NM}\sumi\sumj\max_{ t\le T}\|X_{ijt}\|^{2q}_\infty\Big]\Big)^{2/q}\\
  \le&  (NM)^{2/q} B_{NM}^4.
\end{align*}
Thus, with probability at least $1-C(\log (NM))^{-1}$,
\begin{align*}
\max_{t\in [T]}\max_{l\in[k_0]}\Big\|\frac{1}{NM}\sum_{i,j}( X^2_{ijt,l}  - EX^2_{ijt,l}) \Big\|_\infty \lesssim \sqrt{\frac{\log(k_0^2\vee (NM)  )}{NM}} +  \frac{B^2_{NM}\log(k_0^2\vee (NM)  )}{(NM)^{1-1/q}} \lesssim  \sqrt{\frac{\log a}{NM} }.
\end{align*}
by Lemma \ref{lemma:concentration_inequality_CCK}.
Similarly, with probability at least $1-C((\log N\wedge M))^{-1}$,
\begin{align*}
\max_{t\in [T]}\max_{l\in\{k_0+1,...,k_0+N_0\}}\Big\|\frac{1}{M}\sum_{j}( D^2_{1,ijt,l}  - ED^2_{1,ijt,l}) \Big\|_\infty \lesssim \sqrt{\frac{\log M  }{M}} +  \frac{\log M  }{M} \lesssim  \sqrt{\frac{\log a}{M} },\\
\max_{t\in [T]}\max_{l\in\{k_0+N_0+1,...,k_0+N_0+M_0\}}\Big\|\frac{1}{N}\sum_i( D^2_{2,ijt,l}  - ED^2_{2,ijt,l}) \Big\|_\infty \lesssim \sqrt{\frac{\log N  }{N}} +  \frac{\log N  }{N} \lesssim  \sqrt{\frac{\log a}{N} }
\end{align*}
by Assumption \ref{a:moments} (2).
Thus, with probability at least $1-C(\log(N \wedge M))^{-1}$,
\begin{align*}
\|\bar \Psi - \Psi\|_\infty = O_p\Big( \sqrt{\frac{\log a}{N\wedge M} }\Big).
\end{align*}
Since following Assumption \ref{a:Theta}(4), $s_l\sqrt{\frac{\log a}{N\wedge M}}=o(1) $, we therefore have
\begin{align*}
 \|\bar \Psi - \Psi\|_\infty\|Q_{-l}(\widehat \phi^l -\phi^l )\|_1^2
=&
  O_p\Big( \sqrt{\frac{\log a}{N\wedge M} }\Big) O_p\Big(\frac{s^2_l \log a}{NM}\Big)\\
 =& O_p\Big( s_l\sqrt{\frac{\log a}{N\wedge M} }\Big) O_p\Big(\frac{s_l \log a}{NM}\Big)=o_p\Big(\frac{s_l \log a}{NM}\Big)
\end{align*}
uniformly in $l \in [k_0]$.
Substitute into (\ref{eq:psi_l2}), we obtain
\begin{align*}
  (\widehat \phi^l -\phi^l )'Q_{-l}\Psi_{-l,-l} Q_{-l}(\widehat \phi^l -\phi^l )=O_p\Big(\frac{s \log a}{NM}\Big)
\end{align*}
uniformly in $l \in [k_0]$.
Since under Assumption \ref{a:Theta}(3), $\Lambda_{\min}(\Psi)>0$ and 
\begin{align*}
  \Lambda_{\min}(\Psi)\|\widehat Q_{-l}(\phi^l -\phi^l)\|^2 \le \max_{l\in [k_0]} (\widehat \phi^l -\phi^l )'Q_{-l} \Psi_{-l,-l} Q_{-l}(\widehat \phi^l -\phi^l )
\end{align*}
uniformly in $l \in [k_0]$, we conclude that
\begin{align*}
\|Q_{-l}(\widehat \phi^l -\phi^l)\|=O_p\Big(\sqrt{\frac{s_l \log a}{NM}}\Big),
\end{align*}
uniformly in $l \in [k_0]$.
Thus, the triangle inequality and definition of $Q$ together imply it holds uniformly over $[k_0]$
\begin{align*}
 \|\widehat \phi^l -\phi^l\|=O_p\Big(\sqrt{\frac{s_l \log a}{N\wedge M}}\Big).
\end{align*}
\noindent 
\textbf{Step 2}
We next show $\max_{l \in [k_0]}|\widehat \tau_l^2- \tau_l^2|$. 
By the definition of $\widehat \tau_l$ and the K.K.T. condition, we get the following equality using the decomposition $Z_l=Z_{-l}\phi^l +r_l + \zeta_l$.
\begin{align*}
\widehat \tau_l^2=& \frac{(Z_l - Z_{-l}\widehat \phi^l )'Z_{l}}{NM}\\
=& \frac{[r_l + \zeta_l   - Z_{-l}(\widehat \phi^l -\phi^l)]'(Z_{-l}\phi^l +r_l +\zeta_l )}{NM}\\
=& \frac{\zeta_l'\zeta_l}{NM} + \frac{\zeta_l'Z_{-l}\phi^l}{NM} - \frac{(\widehat \phi^l - \phi^l)'Z_{-l}'Z_{-l}\phi^l}{NM}- \frac{(\widehat \phi^l - \phi^l)'Z_{-l}'\zeta_l}{NM}\\
&
+\Big(\frac{r_l'r_l}{NM}
+
\frac{r'_l Z_{-l}\phi^l}{NM}
+
\frac{2r'_l \zeta_l}{NM}-\frac{(\widehat \phi^l -\phi^l)'Z_{-l}r_l}{NM}
\Big)
\end{align*}
Thus,
\begin{align}
&\max_{l\in [k_0]}|\widehat \tau_l^2-\tau_l^2|
\nonumber\\
\leq&
\max_{l\in [k_0]} \Big|\frac{\zeta_l'\zeta_l}{NM}-\tau_l^2\Big| 
+ 
\max_{l\in [k_0]}\Big|\frac{\zeta_l'Z_{-l}\phi^l}{NM}\Big| 
+
\max_{l\in [k_0]}\Big| \frac{(\widehat \phi^l - \phi^l)'Z_{-l}'Z_{-l}\phi^l}{NM}\Big|
+
\max_{l\in [k_0]}\Big| \frac{(\widehat \phi^l - \phi^l)'Z_{-l}'\zeta_l}{NM}\Big|\nonumber\\
&+
\max_{l\in [k_0]}\Big|\frac{r_l'r_l}{NM}
+
\frac{r'_l Z_{-l}\phi^l}{NM}
+
\frac{2r'_l \zeta_l}{NM}-\frac{(\widehat \phi^l -\phi^l)'Z_{-l} r_l}{NM}
\Big|=(i)+(ii)+(iii)+(iv)+(v). \label{eq:tau2}
\end{align}
It suffices to find bounds for each of the five terms in the last expression. 

First we consider $(i)$.
Under Assumption \ref{a:moments} (1), we have 
\begin{align*}
E\Big[\max_{i,j,t}|\zeta^l_{ijt}|^4\Big]
=&\Big(E\Big[\max_{i,j,t}|\zeta^l_{ijt}|^{2q}\Big]\Big)^{2/q}\\
=&(NM)^{2/q}\Big(E\Big[\frac{1}{NM}\sum_{i,j}\max_{t}|\zeta_{ijt}^l|^{2q}\Big]\Big)^{2/q}\\
  \lesssim&  (NM)^{2/q}\Big( E\Big[\frac{1}{NM}\sumi\sumj\max_{ t\le T}\|X_{ijt}\|^{2q}_\infty\Big]\Big)^{2/q}
  \le  (NM)^{2/q} B_{NM}^4
\end{align*}
for all $l \in[k_0]$.
Therefore, by Lemma  \ref{lemma:concentration_inequality_CCK} and Assumption \ref{a:moments} (1),
\begin{align*}
\max_{l\in[k_0]}\Big\|\frac{1}{NM}\sum_{i,j,t}[ (\zeta^l)^2_{ijt}  - E(\zeta^l)^2_{ijt}] \Big\|_\infty 
\le& 
T \max_{t\in[T]}\max_{l\in[k_0]}\Big\|\frac{1}{NM}\sum_{i,j}[ (\zeta^l)^2_{ijt}  - E(\zeta^l)^2_{ijt}] \Big\|_\infty\\
 \lesssim&
  \sqrt{\frac{\log(k_0\vee (NM)  )}{NM}} +  \frac{B^2_{NM}\log(k_0\vee (NM)  )}{(NM)^{1-1/q}} \lesssim  \sqrt{\frac{\log a}{NM} }
\end{align*}
with probability at least $1-C(\log (NM))^{-1}$.
It follows that
\begin{align*}
\max_{l\in [k_0]} \Big|\frac{\zeta_l'\zeta_l}{NM}-\tau_l^2\Big| =O_p\Big( \sqrt{\frac{\log a}{NM} }\Big).
\end{align*}

Second, we consider $(iv)$ in (\ref{eq:tau2}).
By the regularized events established in Step 1, we have
\begin{align*}
\Big\| \frac{Q^{-1}_{-l}Z_{-l}'\zeta_l}{NM}\Big\|_\infty =\max\Big\{\Big\|\frac{X_{-l}'\zeta_l}{NM}\Big\|_\infty,
\sqrt{N}\Big\|\frac{D_{1}'\zeta_l}{NM}\Big\|_\infty,
\sqrt{M}\Big\|\frac{D_{2}'\zeta_l}{NM}\Big\|_\infty\Big\}=O_p\Big(\sqrt{ \frac{\log a}{NM} }\Big).
\end{align*}
Thus, by (\ref{eq:phi_l1_norm_weighted}),
\begin{align*}
\max_{l\in [k_0]}\Big| \frac{(\widehat \phi^l - \phi^l)'Z_{-l}'\zeta_l}{NM}\Big| \le& 
\Big\| \frac{Q^{-1}_{-l}Z_{-l}'\zeta_l}{NM}\Big\|_\infty \|Q_{-l}(\widehat \phi^l - \phi^l)\|_1=O_p\Big( \frac{s_l\log p}{NM} \Big)=O_p\Big( \sqrt{\frac{s_l\log p}{NM} }\Big)
\end{align*}
follows.

We next consider $(ii)$ in (\ref{eq:tau2}). 
Note that 
$\|\phi^l\|_1=O(\sqrt{s_l})$ in Assumption \ref{a:Theta}(1) implies $\|Q_{-l}\phi^l\|_1=O(\sqrt{s_l})$.
Combining this and the regularized events as before, we have
\begin{align*}
\max_{l\in [k_0]}\Big|\frac{\zeta_l'Z_{-l}\phi^l}{NM}\Big| \le \max_{l\in [k_0]}\Big\|\frac{\zeta_l'Z_{-l}Q^{-1}_{-l}}{NM}\Big\|_\infty \|Q_{-l}\phi^l\|_1=O_p\Big( \sqrt{\frac{s_l\log a}{NM}}\Big).
\end{align*}

Now, we consider $(iii)$ in (\ref{eq:tau2}). 
Using Assumptions \ref{a:sparse_eigenvalues} and \ref{a:Theta} (1) and (3), the fact that $Q^{-1}=\sqrt{NM}S^{-1}$, and the definition of $\bar \Psi$, we obtain
\begin{align*}
\|Z_{-l}\phi^l\|\le& \|Q_{-l}\phi^l\|\max_{\substack{
\|\xi\|=1\\
\|\xi\|_0\le s_l 
}
}
\sqrt{\xi'Q_{-l}^{-1} Z'_{-l}Z_{-l} Q_{-l}^{-1}\xi}\\
\le& 
\sqrt{NM}\|Q_{-l}\phi^l\|\max_{\substack{
\|\xi\|=1\\
\|\xi\|_0\le s_l 
}
}
\sqrt{\xi'\bar \Psi\xi}\\
\le&\sqrt{NM}\cdot O(1)\cdot\sqrt{\varphi_{\max}(\bar \Psi,s_l)}=O_p(\sqrt{NM})
\end{align*}
Furthermore, (\ref{eq:phi_prediction_norm}) implies,
\begin{align*}
\frac{1}{NM}\|Z_{-l}(\widehat \phi^l - \phi^l)\|\le
O_p\Big(
\frac{(s_l\log a)^{1/2}}{NM}
\Big).
\end{align*}
Combining these two intermediate results, we have
\begin{align*}
\max_{l\in [k_0]}\Big| \frac{(\widehat \phi^l - \phi^l)'Z_{-l}'Z_{-l}\phi^l}{NM}\Big|
\le 
\max_{l\in [k_0]} \frac{\|Z_{-l}(\widehat \phi^l - \phi^l)\|\|Z_{-l}\phi^l\|}{NM}=O_p\Big(\sqrt{\frac{s_l \log a}{NM}}\Big),
\end{align*}

Finally, we consider the remaining terms in (\ref{eq:tau2}) that involve $r_l$. 
Note that 
\begin{align*}
\frac{| r'_l r_l|}{NM}\lesssim&\frac{s_l}{NM},\\
\frac{| r'_l \zeta_l|}{NM}\le& \frac{1}{NM}\|r_l\|\|\zeta_l\|\le \frac{1}{\sqrt{NM}} \sqrt{s_l} O_p\Big(\sqrt{\max_{t\in[T]}\max_{l\in[k_0]}\frac{T}{NM}\sumi\sumj E(\zeta^{l}_{ijt})^2}\Big)\lesssim O_p\Big(\sqrt{\frac{s_l}{NM}}\Big)
\end{align*}
follows from Assumption \ref{a:Theta} (2) and Assumption \ref{a:moments} (3) .
Also, under Assumptions \ref{a:sparse_eigenvalues} and \ref{a:Theta} (1) and (2) 
\begin{align*}
\frac{|r_l'Z_l\phi^l|}{NM}\le& \frac{\|r_l\|\|Z_{-l}\phi^l\|}{NM}\\
\le
&\frac{1}{NM}\sqrt{s_l}\|Q_{-l}\phi^l\|\max_{\substack
{\|\delta\|=1\\
\|\delta\|_0\le s_l}
}\sqrt{\delta'Q^{-1}_{-l}Z'_{-l}Z_{-l}Q^{-1}_{-l}\delta}\\
\le
&\frac{1}{\sqrt{NM}}\sqrt{s_l}O(1)\max_{\substack
{\|\delta\|=1\\
\|\delta\|_0\le s_l}
}\sqrt{\delta'\bar\Psi\delta}\\
\le& \sqrt{\frac{s_l}{NM}}\sqrt{\varphi_{\max}(\bar \Psi,s_l)}=O\Big( \sqrt{\frac{s_l}{NM}}\Big)
\end{align*}
with probability at least $1-o(1)$.
A similar argument under Assumptions \ref{a:sparse_eigenvalues} and \ref{a:Theta} (1) and (2) shows that $\frac{|r_l'Z_l(\widehat \phi^l-\phi^l)|}{NM}= O\Big( \sqrt{\frac{s_l}{NM}}\Big)$.

Combining all the results above, we obtain
\begin{align*}
|\widehat \tau_l^2 - \tau_l^2| =O_p\Big(
\sqrt{
\frac{s_l\log a}{NM}
}\Big).
\end{align*} 
uniformly over $[k_0]$
\bigskip\\
\textbf{Step 3}
Since $l\in[k_0]$,
\begin{align}
\frac{1}{ \tau_l^2}= \Theta_{l,l} = Q^{-1}_l \Theta_{l,l} Q^{-1}_l \le \Lambda_{\max}( Q^{-1}\Theta  Q^{-1})= \Lambda_{\max}(\Psi^{-1}) =1/\Lambda_{\min}(\Psi)=O(1),\label{eq:one_over_tau2_eigenvalue}
\end{align}
hods for each $(N,M)$ under Assumption \ref{a:Theta} (3),
where the first inequality follows from the discussion following (B.30) in the Proof of Theorem 1 in \citet{CanerKock2018}.
Therefore, $\widehat\tau^2_l  $ is bounded away from zero in probability, 
and we have
\begin{align}
\max_{l\in[k_0]}\Big|\frac{1}{\widehat \tau_l^2 }- \frac{1}{\tau_l^2}\Big| =O_p\Big(
\sqrt{
\frac{s_l\log a}{NM}
}\Big)\label{eq:one_over_tau2}
\end{align}
by Step 2.

Now, we bound $\max_{l\in [k_0]}\|\widehat \Theta_l - \Theta_l\|_1$. 
Since $\|\phi^l\|_1=O(\sqrt{s_l})$ under Assumption \ref{a:Theta} (1), we have
\begin{align*}
\max_{l\in [k_0]}\|\widehat \Theta_{l} - \Theta_{l}\|_1
\le&
\max_{l\in [k_0]}\Big\|\frac{\widehat C_l}{\widehat \tau_l^2 }- \frac{C_l}{\tau_l^2}\Big\|_1
\\
\le& \max_{l\in [k_0]}\Big|\frac{1}{\widehat \tau_l^2 }- \frac{1}{\tau_l^2}\Big| 
+
\max_{l\in [k_0]}\Big\|\frac{\widehat \phi^l}{\widehat \tau_l^2 }-\frac{\phi^l}{\widehat \tau_l^2 }+ \frac{ \phi^l}{\widehat \tau_l^2 }- \frac{\phi^l}{\tau_l^2}\Big\|_1\\
\le & \max_{l\in [k_0]}\Big|\frac{1}{\widehat \tau_l^2 }- \frac{1}{\tau_l^2}\Big| 
+\max_{l\in[k_0]}\frac{\|\widehat \phi^l - \phi^l\|_1}{\widehat \tau^2_l}
+\max_{l\in[k_0]}\|\phi^l\|_1 \max_{l\in[k_0]}
\Big(\Big|\frac{1}{\widehat \tau_l^2 }- \frac{1}{\tau_l^2}\Big|\Big).
\end{align*}
The first and the third terms can be bounded by (\ref{eq:one_over_tau2}) and the second term can be bounded by (\ref{eq:phi_l1_norm}).
Therefore,
\begin{align*}
\max_{l\in [k_0]}\|\widehat \Theta_{l} - \Theta_{l}\|_1=&O_p\Big(\frac{s_l\log a}{NM}\Big)+O_p\Big(s_l\sqrt{\frac{\log a}{N\wedge M}}\Big)+O_p\Big(s_l \sqrt{\frac{\log a}{NM}}\Big)\\
=&O_p\Big(s_l
\sqrt{\frac{\log a}{N\wedge M}
}\Big).
\end{align*}
Similar lines of argument under Assumption \ref{a:Theta} (1) and $\|\widehat \phi^l - \phi^l\|$ from Step 1 lead to 
\begin{align*}
\|\widehat \Theta_{l} - \Theta_{l}\|=&O_p\Big(
\sqrt{\frac{s_l\log a}{N\wedge M}
}\Big).
\end{align*}
Since $\|\Theta_{l}\|_1 \le \max_{l\in[k_0]}\frac{1}{ \tau^2_l} + \max_{l\in[k_0]}\| \frac{\phi^l}{\tau^2_l}\|_1=O(\sqrt{s_l})$ by (\ref{eq:one_over_tau2_eigenvalue}) and Assumption \ref{a:Theta} (1), it follows that $\|\widehat\Theta_{l}\|_1=O_p(\sqrt{s_l})$ for all $l\in[k_0]$.
\end{proof}

\subsection{Sufficiency for Assumption \ref{a:asymptotic_normality} (i)}
\label{sec:sufficiency_for_i}

\begin{lemma}\label{lemma:approximation_error_Delta}
If Assumptions \ref{a:sparsity}, \ref{a:moments}, \ref{a:sparse_eigenvalues}, and \ref{a:Theta} are satisfied, then
$$
\max_{l\in[k_0]} \abs{ \sqrt{NM}(\widehat\Theta_l' Q\bar\Psi Q  -e_l')(\widehat {\boldsymbol{\eta}}  - \overline{\boldsymbol{\eta}} ) } = o_p(1).
$$
\end{lemma}

\begin{proof}
Recall $\bar\Psi=S^{-1}Z'ZS^{-1}$ and $Q=S/\sqrt{NM}$. 
Also, if we let $\Gamma=ZS^{-1}$, then $\bar \Psi=\Gamma'\Gamma$ and
\begin{align*}
\frac{Z'Z}{NM}=Q\Gamma'\Gamma Q=Q\bar \Psi Q.
\end{align*}
Since $l\in[k_0]$, $Q_{ll}=1$. 
Let $\bar \Psi_l $ denote the $l-$th column of $\bar \Psi$. 
Using the K.K.T. condition for the nodewise lasso, we have
\begin{align}
1=&\frac{(Z_l - Z_{-l}\widehat \phi^l)'Z_l}{\widehat \tau_l^2 NM}=\frac{\widehat \Theta'_{l}Z'Z_l}{NM}=\widehat \Theta_l' Q \Gamma'\Gamma_l\cdot 1=\widehat \Theta_l' Q \bar \Psi_{l}  Q_{ll}.
\label{eq:bound_for_Delta_1}
\end{align}
Also using the K.K.T. condition, we have
\begin{align*}
\frac{Q_{-l}\widehat \kappa_l}{NM} =& \frac{Z_{-l}'(Z_l - Z_{-l}\widehat \phi^l)}{\mu_{\text{node},l}NM}.
\end{align*}
Using the property of the sub-gradient $\kappa_l$, we have
\begin{align*}
\Big\| \frac{Z_{-l}'(Z_l - Z_{-l}\widehat \phi^l)}{\mu_{\text{node},l}NM}
\Big\|_\infty\le \frac{\|\widehat\kappa_l\|_\infty}{NM} \le \frac{1}{NM},
\end{align*}
which is the same as
\begin{align*}
\Big\| \frac{Z_{-l}'Z\widehat C_l}{NM}
\Big\|_\infty\le  \frac{\mu_{\text{node},l}}{NM}
\end{align*}
since $Z_{l}-Z_{-l}\widehat \phi^l=Z \widehat C_l$.  
Divide both sides by $\widehat \tau^2_l$ by using $\widehat \Theta_l =\widehat C_l/\widehat \tau_l^2$ to obtain
\begin{align*}
\Big\| \frac{Z_{-l}'Z\widehat \Theta_{l}}{NM}\Big\|_\infty\le \frac{\mu_{\text{node},l}}{\widehat \tau_l^2NM}
\end{align*}
With some rewriting
\begin{align}
\frac{\mu_{\text{node}}}{\widehat \tau_l^2NM}
\geq
\Big\| \frac{Z_{-l}'Z\widehat \Theta_{l}}{NM}\Big\|_\infty
=
\Big\| \frac{S_{-l}\Gamma_{-l}'\Gamma S\widehat \Theta_{l}}{NM}\Big\|_\infty
=\Big\|\frac{
 S_{-l}\bar \Psi_{-l} S \widehat\Theta_l}{NM}
\Big\|_\infty
=
\Big\| Q_{-l}\bar \Psi_{-l} Q\widehat\Theta_l\Big\|_\infty,
\label{eq:bound_for_Delta_2}
\end{align}
where $S_{-l}$ is $S$ with both the $l$-th column and the $l$-th row removed. 
$Q_{-l}$ is defined similarly.
Applying Lemma \ref{lemma:Theta_1_way} under Assumptions \ref{a:moments}, \ref{a:sparse_eigenvalues}, and \ref{a:Theta}, we have $1/\widehat \tau_l^2=O_p(1)$.
Therefore, by (\ref{eq:bound_for_Delta_1}) and (\ref{eq:bound_for_Delta_2}),
\begin{align*}
\max_{l\in[k_0]}\Big\|\widehat \Theta_{l}'Q\bar\Psi Q - e_l'\Big\|_\infty
=
\max_{l\in[k_0]}\Big\| \frac{Z_l'X\widehat \Theta_{l}}{NM}\Big\|_\infty
\lesssim \max_{l\in[k_0]}\frac{\mu_{\text{node}}}{\widehat \tau_l^2NM}=O_p\Big(\sqrt{\frac{\log a}{NM}}\Big).
\end{align*}
Finally, Lemma \ref{lemma:rates_eta} and Lemma \ref{lemma:regularized_events} with $\mu=C\sqrt{(NM)\log a}$ under Assumptions \ref{a:sparsity}, \ref{a:moments}, and \ref{a:sparse_eigenvalues} together imply
\begin{align*}
&\max_{l\in[p]}|\sqrt{NM}( \widehat\Theta_l
Q\bar\Psi Q- e'_l )(\widehat {\boldsymbol{\eta}}  -\overline{\boldsymbol{\eta}} )|\le 
\sqrt{NM}\max_{l\in[k_0]}\Big\| \widehat \Theta_{l}'Q\bar\Psi Q - e_l'\Big\|_\infty \|\widehat {\boldsymbol{\eta}} - \overline{\boldsymbol{\eta}}\|_1\\
=&\sqrt{NM}O_p\Big(\sqrt{\frac{ \log a}{NM}}\Big)O_p\Big(s\sqrt{\frac{ \log a}{N\wedge M}}\Big)=O_p\Big(\sqrt{\frac{s^2 (\log a)^2}{N\wedge M}}\Big)=o_p(1)
\end{align*}
as claimed.\footnote{Note that Lemma \ref{lemma:rates_eta}, as it is stated, requires Assumptions \ref{a:sparsity}, \ref{a:rates}, \ref{a:penalty_loading}, and \ref{a:restricted_eigenvalue}. While Assumption \ref{a:sparsity} is directly invoked by the statement of Lemma \ref{lemma:approximation_error_Delta}, Assumption \ref{a:rates} is implied by Assumption \ref{a:moments} through Lemma \ref{lemma:regularized_events}, Assumption \ref{a:penalty_loading} is trivially satisfied under the current setting with $\widehat\Upsilon=I$, and Assumption \ref{a:restricted_eigenvalue} is implied by Assumption \ref{a:sparse_eigenvalues}.}
\end{proof}

\subsection{Sufficiency for Assumption \ref{a:asymptotic_normality} (ii)}
\label{sec:sufficiency_for_ii}

\begin{lemma}\label{lemma:suff_2}
Suppose that Assumptions \ref{a:sparsity}, \ref{a:moments}, \ref{a:sparse_eigenvalues}, and \ref{a:Theta} are satisfied.
Then,
$$
\max_{l\in[k_0]} \abs{ \widehat \Theta_l' Z'R/\sqrt{NM} }=o_p(1).
$$
\end{lemma}

\begin{proof}
Note that 
$$
\max_{l\in[k_0]} \| \widehat \Theta_l \| = O_p(1)
$$
by Assumption \ref{a:Theta} (1) and (4) and Lemma \ref{lemma:Theta_1_way} under Assumptions \ref{a:moments}, \ref{a:sparse_eigenvalues}, and \ref{a:Theta}.
Therefore, 
\begin{align*}
\max_{l\in[k_0]} \abs{ \widehat \Theta_l' Z'R/\sqrt{NM} }
\leq
\frac{1}{NM} \max_{l\in[k_0]} \| \widehat \Theta_l \| \| Z'R \| = o_p(1)
\end{align*}
follows under Assumption \ref{a:sparsity} (3).
\end{proof}

\subsection{Sufficiency for Assumption \ref{a:asymptotic_normality} (iii)}
\label{sec:sufficiency_for_iii}

\begin{lemma}\label{lemma:suff_3}
Suppose that Assumptions \ref{a:moments}, \ref{a:sparse_eigenvalues}, \ref{a:Theta} and \ref{a:variance} are satisfied.
Then,
\begin{align*}
V_{ll}^{-1/2}\widehat\Theta_l' Z'\varepsilon /\sqrt{NM} \leadsto N(0,1).
\end{align*}
\end{lemma}

\begin{proof}
First we show $\frac{1}{\sqrt{NM}}\Theta_l' Z'\varepsilon \leadsto N(0,V_{ll})$. 
Note that we have 
$$
E[ \frac{1}{\sqrt{NM}}\Theta_l' Z'\varepsilon]=\frac{1}{\sqrt{NM}}E[\Theta_l' Z'E[\varepsilon|Z]]=0
$$ 
and 
\begin{align*}
V_{ll}=E\Big[\Big(\frac{1}{\sqrt{NM}}\Theta_l' Z'\varepsilon\Big)\Big(\frac{1}{\sqrt{NM}}\Theta_l' Z'\varepsilon\Big)'\Big]=\Theta_l' \Omega \Theta_l\ge\underline k>0
\end{align*}
under Assumption \ref{a:variance}.
Furthermore, by Assumption \ref{a:moments}
\begin{align*}
E\Big| \frac{1}{\sqrt{NM}}\Theta_l' Z'\varepsilon\Big|^{q}
\le &
 \frac{1}{(NM)^{q/2}}E\|\Theta_l\|^q_1\max_{k\in \supp(\Theta_l)}\sumi\sumj\sumt \Big|Z_{ijt,k}\varepsilon_{ijt}\Big|^{q}\\
\lesssim
&\frac{ s_l^{q/2}}{(NM)^{q/2}}\sum_{k\in \supp(\Theta_l)}E\sumi\sumj\sumt \Big|Z_{ijt,k}\varepsilon_{ijt}\Big|^{q}\\
\le
&\frac{ s_l^{q/2}\cdot s_l}{(NM)^{q/2}}\max_{l\in [p]}E\sumi\sumj \sumt \Big|Z_{ijt,k}\varepsilon_{ijt}\Big|^{q}\\
\le&
\frac{ s_l^{q/2+1}(NM)}{(NM)^{q/2}}\max_{l\in [p]}\sqrt{\frac{1}{NM}E\sumi\sumj\sumt \Big|Z_{ijt,k}\Big|^{2q} \frac{1}{NM} E\sumi\sumj\sumt \Big|\varepsilon_{ijt}\Big|^{2q}
}\\
\le& \frac{ s_l^{q/2+1}}{(NM)^{q/2-1}}O(1)=o(1),
\end{align*} 
where $q> 4$, 
the first inequality follows from a dual norm inequality, 
the second and the third from the fact that $\|\Theta_l\|_1\lesssim \sqrt{s_l}$ and $\|\Theta_l\|_0\le s_l$ implied by Assumption \ref{a:Theta}(1), 
the fourth from Cauchy-Schwartz's inequality, and
the fifth from Assumption \ref{a:moments} and the last equality follows from Assumption \ref{a:Theta} (4). 
This verifies the Lyapunov's condition. 
Thus, we have $\frac{1}{\sqrt{NM}}\Theta_l' Z'\varepsilon \leadsto N(0,V_{ll})$. 

Now, we show $|\frac{1}{\sqrt{NM}}(\widehat\Theta_l-\Theta_l)'Z'\varepsilon|=o_p(1)$. 
Invoking Lemmata \ref{lemma:regularized_events} and \ref{lemma:Theta_1_way} under Assumptions \ref{a:moments}, \ref{a:sparse_eigenvalues}, \ref{a:Theta}, we have
\begin{align*}
\Big|\frac{1}{\sqrt{NM}}(\widehat\Theta_l-\Theta_l)'Z'\varepsilon\Big|
\le&
\|\widehat\Theta_l-\Theta_l\|_1\Big\|\frac{1}{\sqrt{NM}}Z'\varepsilon\Big\|_\infty\\
\le&
 O_p\Big(\sqrt{\frac{s_l^2\log a}{N\wedge M}}\Big)
 O_p\Big(\sqrt{\log a}\Big)\\
=& O_p\Big(\sqrt{\frac{s_l^2(\log a)^2}{N\wedge M}}\Big)=o_p(1).
\end{align*}
Combining these results concludes $\frac{1}{\sqrt{NM}}\widehat\Theta_l' Z'\varepsilon \leadsto N(0,V_{ll})$.
\end{proof}

\subsection{Empirical Pre-Sparsity}
\label{sec:empirical_pre_sparsity}
The following lemma is a minor modification of Lemma 8 in \citet{BelloniChenChernozhukovHansen2012}.
\begin{lemma}[Empirical Pre-sparsity]\label{lemma:empirical_pre-sparsity}
If Assumptions \ref{a:sparsity}, \ref{a:moments}, \ref{a:sparse_eigenvalues}, and \ref{a:Theta} are satisfied, then we have
\begin{align*}
\widehat s_l =O_p(s_l) \text{ and } \widehat s =O_p(s),
\end{align*}
where 
$\widehat s_l = \| \widehat{\phi}^l \|_0$ and
$\widehat s = \| \widehat{\boldsymbol{\eta}} \|_0$.
\end{lemma}
\begin{proof}
Let $\hat m_l= \abs{\hat  T_l \setminus T_l}$, where $T_l=\supp(\phi^l)$ and $\hat T_l=\supp(\hat \phi^l)$.
From K.K.T. condition, we have
\begin{align*}
2 (Q^{-1}_{-l} Z'_{-l}(Z_l-Z_{-l}\hat \phi^l))_k=\mu_{\text{node},l}\cdot \text{sign}(\hat \phi^l_k)
\end{align*}
for all $l\in [k_0]$ and $k\in \hat T_l \setminus T_l$.
Thus,
\begin{align}
\mu_{\text{node},l} \sqrt{\hat m_l} \leq &
 2\| ( Q^{-1}_{-l}Z'_{-l}(Z_l-Z_{-l}\hat \phi^l))_{\hat  T_l \setminus T_l}\|
 +
 2\| (Q^{-1}_{-l} Z'_{-l}r_l)_{\hat  T_l \setminus T_l}\|\nonumber\\
& +
 2\| (Q^{-1}_{-l} Z'_{-l}Z_{-l}(\hat \phi^l -\phi^l))_{\hat  T_l \setminus T_l}\|\nonumber\\
=& (1) +(2)+(3). \label{eq:pre-sparsity}
\end{align}
We bound the three terms in the last expression separately. 
First, Lemma \ref{lemma:regularized_events} under Assumption \ref{a:moments} yields
\begin{align*}
(1)
\le & 2\sqrt{\hat m_l}\|Q_{-l}^{-1} Z'_{l} \zeta_l\|_\infty
\le \sqrt{\hat m_l}\frac{\mu_{node,l}}{c}
\end{align*}
with probability at least $1-C(\log(N \wedge M))^{-1}$.
Second, 
\begin{align*}
\| (Q^{-1}_{-l} Z'_{-l}r_l)_{\hat  T_l \setminus T_l}\|
=&\sup_{
\substack{\|\delta\|=1\\
\|\delta\|_0\le \hat m_l
}}
|\delta'Q^{-1}_{-l} Z'_{-l}r_l|  \\
\le& 
\sup_{
\substack{\|\delta\|=1\\
\|\delta\|_0\le \hat m_l
}}
\|\delta'Q^{-1}_{-l}Z'_{-l}\|\| r_l\|\\
\le&
\sup_{
\substack{\|\delta\|=1\\
\|\delta\|_0\le \hat m_l
}}
(NM)
\sqrt{\delta'\bar \Psi \delta} \sqrt{\frac{s_l}{NM}}\\
\le& (NM) \sqrt{\varphi_{\max}(\bar \Psi,\hat m_l)} \sqrt{\frac{s_l}{NM}}
\end{align*}
follows by Assumptions \ref{a:sparse_eigenvalues} and \ref{a:Theta}.
Therefore,
\begin{align*}
(2)\le 2 (NM)\sqrt{ \varphi_{\max}(\bar \Psi,\hat m_l)} \sqrt{\frac{s_l}{NM}}.
\end{align*}
Finally, by Lemma \ref{lemma:Theta_1_way} under Assumptions \ref{a:moments}, \ref{a:sparse_eigenvalues}, \ref{a:Theta}, we obtain
\begin{align*}
\|( Q^{-1}_{-l} Z'_{-l}Z_{-l}(\hat \phi^l -\phi^l))_{\hat  T_l \setminus T_l}\|
\le&
\sup_{
\substack{\|\delta\|=1\\
\|\delta\|_0\le \hat m_l
}}
|\delta'Q^{-1}_{-l} Z'_{-l}Z_{-l}(\hat \phi^l -\phi^l)|\\
\le& 
\sup_{
\substack{\|\delta\|=1\\
\|\delta\|_0\le \hat m_l
}}
\|\delta'Q^{-1}_{-l} Z'_{-l}\|\|Z_{-l}(\hat \phi^l -\phi^l)\|\\
\le&   (NM)\sqrt{\varphi_{\max}(\bar \Psi,\hat m_l)} \sqrt{\frac{s_l\log a}{NM}}
\end{align*}
with probability at least $1-C(\log(N \wedge M))^{-1}$, where
the last inequality is due to Assumption \ref{a:sparse_eigenvalues} and Lemma \ref{lemma:Theta_1_way}.
Using these bounds and (\ref{eq:pre-sparsity}), we obtain
\begin{align*}
\sqrt{\hat m_l}\lesssim \sqrt{\varphi_{\max}(\bar\Psi,\hat m_l)} \sqrt{s_l}=O(\sqrt{s_l})
\end{align*}
with probability $1-o(1)$. 
Under Assumptions \ref{a:sparsity}, \ref{a:moments}, \ref{a:sparse_eigenvalues}, and \ref{a:Theta}, the result for $\hat s$ can be established following analogous arguments.
\end{proof}

\newpage
\bibliography{biblio}

\clearpage
\thispagestyle{empty}
\begin{table}
\begin{center}
\begin{tabular}{lrrrrrr}\hline\hline
Country&1990&1995&2000&2005&2010&2015\\ \hline
China         &3.08 &5.76&7.22&10.74&13.32&15.67\\
USA           &18.47&12.32&12.62&8.79&8.19&8.82\\
Germany       &6.98 &9.87&8.18&8.84&7.86&7.69\\
Japan         &8.20 &8.95&7.97&6.34&5.46&4.40\\
Korea         &1.86 &2.22&2.70&2.86&3.09&3.38\\
France        &6.55 &5.54&4.59&4.20&3.49&3.20\\
Italy         &5.36 &4.31&3.42&3.31&2.82&2.76\\
Netherlands   &3.93 &3.47&3.03&2.97&2.83&2.61\\
Canada        &1.34 &4.08&4.47&3.48&2.61&2.60\\
United Kingdom&4.78 &4.66&4.36&3.41&2.57&2.56\\ \hline
ROW           &39.47&38.81&41.43&45.07&47.78&46.30\\ 
$N_{ROW}$     &192&214&224&222&224&227\\
ROW/$N_{ROW}$ &0.21&0.18&0.18&0.20&0.21&0.20\\ \hline\hline
\end{tabular}
\end{center}
\caption{World Import Shares Over Time (\%). Notes: The above table reports the share of world imports for the 10 largest importers from the WITS database in selected years.  ROW represents the combined share of all other countries.  $N_ROW$ is the number of countries which comprise ROW and ROW/$N_{ROW}$ is the average import share among the ROW countries.}
\label{tab:world_import_shares_over_time}
\end{table}

\clearpage
\thispagestyle{empty}
\begin{table}
\begin{center}
\begin{tabular}{lrrrrrr}\hline\hline
Country&1990&1995&2000&2005&2010&2015\\ \hline
USA&22.47&14.73&1925&15.95&12.26&13.79\\
China&3.34&6.01&6.13&8.33&11.13&12.35\\
Germany&5.19&9.09&6.64&6.36&5.78&5.45\\
United Kingdom&5.74&5.35&4.72&4.15&3.37&3.39\\
Japan&5.72&5.39&4.90&4.30&3.55&3.30\\
France&6.64&5.56&4.42&4.26&3.63&3.08\\
Netherlands&4.29&3.93&3.19&3.22&3.21&2.84\\
Canada&1.13&3.44&3.59&2.85&2.46&2.47\\
Korea&2.22&2.23&2.01&1.99&2.27&2.39\\
Switzerland&2.38&1.80&2.31&2.30&2.26&2.29\\ \hline
ROW&40.88&42.47&42.84&46.30&50.07&48.65\\ 
$N_{ROW}$     &186&209&218&218&219&222\\
ROW/$N_{ROW}$ &0.22&0.20&0.20&0.21&0.23&0.22\\ \hline\hline
\end{tabular}
\end{center}
\caption{World Export Shares Over Time (\%). Notes: The above table reports the share of world exports for the 10 largest exporters from the WITS database in selected years.  ROW represents the combined share of all other countries.  $N_ROW$ is the number of countries which comprise ROW and ROW/$N_{ROW}$ is the average export share among the ROW countries.}
\label{tab:world_export_shares_over_time}
\end{table}

\clearpage
\thispagestyle{empty}
\begin{table}
	\centering
		\begin{tabular}{rccccccc}
			\hline\hline
			\multicolumn{2}{l}{True Model = (I)} && \multicolumn{3}{c}{Fixed Effect Estimators} && \\
			\cline{4-6}
			\multicolumn{1}{l}{$N=10$ ($NMT=450)$} & OLS && FE-I & FE-II & FE-III && POST\\
			\hline
			Under-Fitting or Over-Fitting & Under && Correct & Over & Over && Robust\\
			\hline
			Average                & 1.466 && 0.996 & 0.996 & 0.996 && 1.066 \\
			Bias                   & 0.466 &&-0.004 &-0.004 &-0.004 && 0.066 \\
			Standard Deviation     & 0.342 && 0.484 & 0.486 & 0.539 && 0.422 \\
			Root Mean Square Error & 0.578 && 0.484 & 0.486 & 0.539 && 0.428 \\
			95\% Coverage          & 0.712 && 0.941 & 0.938 & 0.909 && 0.961 \\
			\hline\hline
			\\
		\end{tabular}

		\begin{tabular}{rccccccc}
			\hline\hline
			\multicolumn{2}{l}{True Model = (II)} && \multicolumn{3}{c}{Fixed Effect Estimators} && \\
			\cline{4-6}
			\multicolumn{1}{l}{$N=10$ ($NMT=450)$} & OLS && FE-I & FE-II & FE-III && POST\\
			\hline
			Under-Fitting or Over-Fitting & Under && Under & Correct & Over && Robust\\
			\hline
			Average                & 1.393 && 1.206 & 1.002 & 1.003 && 1.110 \\
			Bias                   & 0.393 && 0.206 & 0.002 & 0.003 && 0.110 \\
			Standard Deviation     & 0.334 && 0.473 & 0.488 & 0.533 && 0.405 \\
			Root Mean Square Error & 0.515 && 0.516 & 0.488 & 0.533 && 0.420 \\
			95\% Coverage          & 0.771 && 0.914 & 0.938 & 0.910 && 0.957 \\
			\hline\hline
			\\
		\end{tabular}
		
		\begin{tabular}{rccccccc}
			\hline\hline
			\multicolumn{2}{l}{True Model = (III)} && \multicolumn{3}{c}{Fixed Effect Estimators} && \\
			\cline{4-6}
			\multicolumn{1}{l}{$N=10$ ($NMT=450)$} & OLS && FE-I & FE-II & FE-III && POST\\
			\hline
			Under-Fitting or Over-Fitting & Under && Under & Under & Correct && Robust\\
			\hline
			Average                & 1.461 && 1.441 & 1.421 & 1.008 && 1.108 \\
			Bias                   & 0.461 && 0.441 & 0.421 & 0.008 && 0.108 \\
			Standard Deviation     & 0.339 && 0.352 & 0.361 & 0.531 && 0.429 \\
			Root Mean Square Error & 0.572 && 0.564 & 0.555 & 0.531 && 0.442 \\
			95\% Coverage          & 0.717 && 0.746 & 0.767 & 0.909 && 0.954 \\
			\hline\hline
		\end{tabular}
	\caption{Monte Carlo simulation results under Model (I) (top panel), Model (II) (middle panel), and Model (III) (bottom panel) with size $N=10$ ($NMT=450$).}
	\label{tab:simulation}
\end{table}

\clearpage
\thispagestyle{empty}
\begin{table}
	\centering
		\begin{tabular}{rccccccc}
			\hline\hline
			\multicolumn{2}{l}{True Model = (I)} && \multicolumn{3}{c}{Fixed Effect Estimators} && \\
			\cline{4-6}
			\multicolumn{1}{l}{$N=15$ ($NMT=1050)$} & OLS && FE-I & FE-II & FE-III && POST\\
			\hline
			Under-Fitting or Over-Fitting & Under && Correct & Over & Over && Robust\\
			\hline
			Average                & 1.385 && 1.002 & 1.002 & 1.000 && 1.006 \\
			Bias                   & 0.385 && 0.002 & 0.002 & 0.000 && 0.006 \\
			Standard Deviation     & 0.215 && 0.316 & 0.317 & 0.336 && 0.274 \\
			Root Mean Square Error & 0.441 && 0.316 & 0.317 & 0.336 && 0.274 \\
			95\% Coverage          & 0.568 && 0.941 & 0.940 & 0.924 && 0.957 \\
			\hline\hline
			\\
		\end{tabular}

		\begin{tabular}{rccccccc}
			\hline\hline
			\multicolumn{2}{l}{True Model = (II)} && \multicolumn{3}{c}{Fixed Effect Estimators} && \\
			\cline{4-6}
			\multicolumn{1}{l}{$N=15$ ($NMT=1050)$} & OLS && FE-I & FE-II & FE-III && POST\\
			\hline
			Under-Fitting or Over-Fitting & Under && Under & Correct & Over && Robust\\
			\hline
			Average                & 1.481 && 1.163 & 1.000 & 1.000 && 1.031 \\
			Bias                   & 0.481 && 0.163 & 0.000 & 0.000 && 0.031 \\
			Standard Deviation     & 0.214 && 0.308 & 0.313 & 0.331 && 0.265 \\
			Root Mean Square Error & 0.526 && 0.348 & 0.313 & 0.331 && 0.267 \\
			95\% Coverage          & 0.412 && 0.910 & 0.944 & 0.930 && 0.959 \\
			\hline\hline
			\\
		\end{tabular}
		
		\begin{tabular}{rccccccc}
			\hline\hline
			\multicolumn{2}{l}{True Model = (III)} && \multicolumn{3}{c}{Fixed Effect Estimators} && \\
			\cline{4-6}
			\multicolumn{1}{l}{$N=15$ ($NMT=1050)$} & OLS && FE-I & FE-II & FE-III && POST\\
			\hline
			Under-Fitting or Over-Fitting & Under && Under & Under & Correct && Robust\\
			\hline
			Average                & 1.409 && 1.383 & 1.367 & 0.998 && 1.055 \\
			Bias                   & 0.409 && 0.383 & 0.367 &-0.002 && 0.055 \\
			Standard Deviation     & 0.222 && 0.231 & 0.235 & 0.334 && 0.282 \\
			Root Mean Square Error & 0.465 && 0.447 & 0.435 & 0.334 && 0.287 \\
			95\% Coverage          & 0.548 && 0.615 & 0.650 & 0.927 && 0.951 \\
			\hline\hline
		\end{tabular}
	\caption{Monte Carlo simulation results under Model (I) (top panel), Model (II) (middle panel), and Model (III) (bottom panel) with size $N=15$ ($NMT=1050$).}
	\label{tab:simulation15}
\end{table}

\clearpage
\thispagestyle{empty}
\begin{table}
	\centering
		\begin{tabular}{rccccccc}
			\hline\hline
			\multicolumn{2}{l}{True Model = (I)} && \multicolumn{3}{c}{Fixed Effect Estimators} && \\
			\cline{4-6}
			\multicolumn{1}{l}{$N=20$ ($NMT=1900)$} & OLS && FE-I & FE-II & FE-III && POST\\
			\hline
			Under-Fitting or Over-Fitting & Under && Correct & Over & Over && Robust\\
			\hline
			Average                & 1.213 && 0.996 & 0.996 & 0.996 && 0.956 \\
			Bias                   & 0.213 &&-0.004 &-0.004 &-0.004 &&-0.044 \\
			Standard Deviation     & 0.162 && 0.228 & 0.229 & 0.239 && 0.199 \\
			Root Mean Square Error & 0.268 && 0.228 & 0.229 & 0.239 && 0.204 \\
			95\% Coverage          & 0.739 && 0.948 & 0.948 & 0.938 && 0.955 \\
			\hline\hline
			\\
		\end{tabular}

		\begin{tabular}{rccccccc}
			\hline\hline
			\multicolumn{2}{l}{True Model = (II)} && \multicolumn{3}{c}{Fixed Effect Estimators} && \\
			\cline{4-6}
			\multicolumn{1}{l}{$N=20$ ($NMT=1900)$} & OLS && FE-I & FE-II & FE-III && POST\\
			\hline
			Under-Fitting or Over-Fitting & Under && Under & Correct & Over && Robust\\
			\hline
			Average                & 1.332 && 1.218 & 0.999 & 1.000 && 1.052 \\
			Bias                   & 0.332 && 0.218 &-0.001 & 0.000 && 0.052 \\
			Standard Deviation     & 0.168 && 0.222 & 0.230 & 0.240 && 0.214 \\
			Root Mean Square Error & 0.371 && 0.311 & 0.230 & 0.240 && 0.221 \\
			95\% Coverage          & 0.497 && 0.834 & 0.945 & 0.935 && 0.954 \\
			\hline\hline
			\\
		\end{tabular}
		
		\begin{tabular}{rccccccc}
			\hline\hline
			\multicolumn{2}{l}{True Model = (III)} && \multicolumn{3}{c}{Fixed Effect Estimators} && \\
			\cline{4-6}
			\multicolumn{1}{l}{$N=20$ ($NMT=1900)$} & OLS && FE-I & FE-II & FE-III && POST\\
			\hline
			Under-Fitting or Over-Fitting & Under && Under & Under & Correct && Robust\\
			\hline
			Average                & 1.390 && 1.374 & 1.340 & 0.995 && 1.030 \\
			Bias                   & 0.390 && 0.374 & 0.340 &-0.005 && 0.030 \\
			Standard Deviation     & 0.168 && 0.174 & 0.180 & 0.244 && 0.213 \\
			Root Mean Square Error & 0.425 && 0.413 & 0.385 & 0.244 && 0.215 \\
			95\% Coverage          & 0.363 && 0.414 & 0.520 & 0.935 && 0.949 \\
			\hline\hline
		\end{tabular}
	\caption{Monte Carlo simulation results under Model (I) (top panel), Model (II) (middle panel), and Model (III) (bottom panel) with size $N=20$ ($NMT=1900$).}
	\label{tab:simulation20}
\end{table}

\end{document}